\documentclass[
preprint,
nofootinbib,
amsmath,amssymb,
]{revtex4-2}

\usepackage{graphicx}
\usepackage{dcolumn}
\usepackage{bm}
\usepackage{braket}
\usepackage{graphicx}
\usepackage{hyperref}
\usepackage{bbold}
\usepackage{tikz}
\usepackage{amsthm}
\usepackage{mathtools} 
\usepackage[shortlabels]{enumitem}
\usepackage{caption}
\usepackage{subcaption}
\usepackage{tikz-cd}
\usepackage{pgfplots}
\pgfplotsset{compat=1.17,colormap/blackwhite}
\usepackage{adjustbox}
\usetikzlibrary{matrix}
\usepackage{tikz-3dplot}
\usepackage[colorinlistoftodos,prependcaption,textsize=small]{todonotes}
\usetikzlibrary{patterns,calc,decorations.markings}

\theoremstyle{plain}
\newtheorem{thm}{Theorem}[section]
\newtheorem{lem}[thm]{Lemma}
\newtheorem{prop}[thm]{Proposition}
\theoremstyle{definition}
\newtheorem{defn}[thm]{Definition}

\theoremstyle{remark}

\begin{document}


\title{A dualization approach to the Ground State Subspace Classification of Abelian Higher Gauge Symmetry Models.}

 \author{J. Lorca Espiro} 
 \email{javier.lorca@ufrontera.cl}
 \altaffiliation[]{Departamento de Ciencias F\'{i}sicas, Universidad de la Frontera, Avda. Francisco Salazar 01145, Casilla 54-D Temuco, Chile.}

\date{\today}

\begin{abstract}
In the literature, abelian higher gauge symmetry models are shown to be valid in all finite dimensions and exhibit the characteristic behavior of SPT phases models. While the ground state degeneracy and the entanglement entropy were thoroughly studied, the classification of the ground state space still remained obscure. Based on differentio-geometric approach and, anticipating the notation of the current paper,  if $\left( C_{\bullet} , \partial^C_{\bullet} \right)$ is the chain complex associated to the geometrical content of these models,  while $\left( G_{\bullet} , \partial^G_{\bullet} \right)$ is its symmetries counterpart, we show that the ground state space is classified by a $H^0 (C,G) \times H_0 (C,G)$ group, where $H^0(C,G)$ is the $0$-th cohomology and $H_0 (C,G)$ is the corresponding $0$-th homology group with coefficients in the chain complex. 
\end{abstract}

\pacs{Valid PACS appear here}
\maketitle
\newpage

\section{\label{sec:intro}Introduction}

Broadly speaking, the standard approach to the topological phases classification problem has been to study Topological Quantum Field Theories (TQFTs) modelled after some specific physically motivated symmetry. For instance, the renowned Dijkgraaf-Witten classification is based on the existence of a global symmetry modelled by a gauge group \cite{Witten,Atiyah88,Dijkgraaf89topologicalgauge}. This approach has lead to the notion of invertible topological phases \cite{Freed:2016rqq,Freed:2019jzd}, from which the symmetry protected topological (SPT) phases \cite{Wen-Gu09,Chen12,Oshikawa10,Oshikawa12,Fidkowsky11,LevinGu,Senthil_2015,Gaiotto_2019} subclass has been thoroughly studied. The latter can be characterized as equivalence classes of gapped systems with a unique $G$-symmetric ground state when embedded on closed manifolds. The classification itself is given by bordism groups (basically cohomology)\cite{Chen-Wen13,Kapustin14,Barkeshli:2021ypb}.

Moreover, the idea of symmetry has also been generalized as to include the so called \textit{higher symmetries} into the classification scheme. Interesting results have been achieved by categorical methods (See \cite{Kong:2022cpy} and the references therein), particularly in models based on $2$-groups and $3$-groups \cite{Kapustin13,Radenkovi__2019,Ritter_2016,Gastel_2019,Zucchini_2017,Carow_Watamura_2016}, where SPT phases with higher symmetries have been developed and studied. Somewhat related to the latter, we build upon the results shown in \cite{dealmeida2017topological,Ibieta_Jimenez_2020}, where a large class of models with higher (abelian) symmetries with SPT phases were presented. The models are a direct generalization of the Toric Code (TC) or, more specifically, the quantum double models (QDM) \cite{Beigi,Hu,Bullivant16,Bullivant17} in its abelian version, with the particularity of being well defined in all (finite countable) dimensions. As them, they describe gapped topological phases of matter with a degenerate ground state subspace. The most relevant results found so far can be summarized as: 
\begin{enumerate}[i]
\item The ground state degeneracy (GSD) is given by the cardinality of the cohomology classes having coefficients in the chain complex of abelian groups;
and
\item The bipartition entanglement entropy satisfy the area law and the sub-leading terms are explicitly dependent on the topology of the entangling manifold.
\end{enumerate}

One of the most attractive features of these models is the fact that it deeply intersects with other fields of study. Aside from the obvious interest of these as toy models from the condensed matter perspective, there is a well established amount of literature that discuss the connection of these models with fault-tolerant quantum computation codes. The link is specially patent in the case of $1$-gauge models embedded in $2D$ spaces, which in our terminology means \textit{having degrees of freedom associated to the links of a triangulated surface}, see, for instance, \cite{Yoshida16,Yoshida15,Yoshida17,Kitaev2,Nayak08,Freedman03,Hastings,CalderbankCSS}. It remains to show, however, if these higher dimensional models have novel effects interesting for quantum computation. A situation that will be elucidated once the study of the excited states is carried out, so far postponed for future works.

It is known that error correction codes can be studied as quantum CSS stabilizer codes \cite{Raussendorf_2006,Briegel_2009,Devitt}. The latter are characterized by the fact that the ground state subspace is used to encode quantum information. Therefore, characterizing and classifying the ground state subspace is paramount to its constructions. In our case, it is immediate that error correction codes are included in these abelian higher gauge symmetry models, and, moreover, the ground state space has been characterized and classified in this paper, opening the door for studies in this direction in the future. As a hint we are headed in the right direction, these higher version models can also be studied in terms of Homology, as it is the case with \cite{Bombin07,Bombin13,Bravyi14}, making these the  higher dimensional versions of the \textit{homological quantum error correction codes} \cite{Bravyi98,Vrana15,Anderson13hom}.

Considering the structure of the models presented (see Sec. \ref{sec:preliminaries}), if $\left( C_{\bullet} , \partial^C_{\bullet} \right)$ is the chain complex associated to the geometrical content, while $\left( G_{\bullet} , \partial^G_{\bullet} \right)$ is the chain complex associated to the symmetry(ies) content, the main result of this paper is establishing that the ground state space is classified by a $H^0 (C,G) \times H_0 (C,G)$ group, where $H^0(C,G)$ is the $0$-th cohomology group of Eq. \ref{pcohomology} and $H_0 (C,G)$ is the corresponding $0$-th homology group of Eq. \ref{phomology}. We hope to have come close in trying to reach a balance between a \textit{rigorously formal} and a \textit{physically heuristic} style in order to not to alienate any of the possible readers interested in these results.

The paper has been organized as follows: Sec. \ref{sec:preliminaries} reviews the mathematical structure of the gauge configuration and gauge representation, as well as an explicit construction of their associated Hilbert spaces ($p$-Hilbert spaces) where the isomorphic nature of them is emphasized; Sec. \ref{sec:QDM} reviews the formalism presented in the aforementioned references, with an eye on the duality between the configuration and representation spaces that will pay off in the following sections; Sec. \ref{sec:topdesc} embarks on a characterization of the ground state subspace via homological and cohomological description, which by the end of the section are shown to be equivalent. An optional discussion on the explicit calculation of the ground state degeneracy for the interested reader; Sec. \ref{sec:classification} encompasses the main result of this paper, which is the classification of the ground state subspace via the classifying space technique; we finish the paper with a short section (Sec. \ref{sec:FinalRemarks}) summing up the results and discussing its possible consequences.

\section{\label{sec:preliminaries} Technical Preliminaries}

Let us consider two chain complexes $\left( C_{\bullet}, \partial^C_{\bullet} \right)$ and $\left( G_{\bullet}, \partial^G_{\bullet} \right)$:
\begin{align*}
\begin{array}{c}
\cdots \xrightarrow{} C_{n} \xrightarrow{\partial_n^C} C_{n-1} \xrightarrow{\partial_{n-1}^C} C_{n-2} \xrightarrow{} \cdots \\
\text{with} \quad \partial^C_{n-1} \circ \partial^C_n =0
\end{array}
\quad \quad \text{and} \quad \quad \begin{array}{c}
\cdots \xrightarrow{} G_{n} \xrightarrow{\partial_n^G} G_{n-1} \xrightarrow{\partial_{n-1}^G} G_{n-2} \xrightarrow{} \cdots \\
\text{with} \quad \partial^G_{n-1} \circ \partial^G_n = 0
\end{array}\quad .
\end{align*}
The notation is rooted in the fact that the first chain is usually related to the \textit{geometrical content}, while the second to the \textit{gauge content} of the models. However, more general models can also be studied within this structure. In a concrete fashion, we usually take the upper chain to be $C = C\left( X \right)$ the triangulation of some compact manifold $X$ with $\partial^C_{\bullet}$ being the usual boundary operator. In FIG. \ref{simplices} we can see the building blocks for the analysis of a surface, regarded as being isomorphic to a subset of $\mathbb{R}^2$. As a visual example, a sketched figure of a triangulated $T^2$ torus is shown in Fig. \ref{torus}.
\begin{figure}[h!]
\centering
\begin{subfigure}{.5\textwidth}
\centering
\begin{tikzpicture}
\fill (-2,0) circle (1.5pt) {};
\node (0) at (-2,-0.5) {$0$-simplex};
\fill (0,0) circle (1.5pt) {};
\fill (0,1) circle (1.5pt) {};
\draw (0,0) -- (0,1) {};
\draw[->] (0,0) -- (0,0.5) {}; 
\node (1) at (0,-0.5) {$1$-simplex};
\fill (1,0) circle (1.5pt) {};
\fill (2.5,0) circle (1.5pt) {};
\fill (1.5,1.5) circle (1.5pt) {};
\path[shade,draw] (1,0) -- (2.5,0) -- (1.5,1.5) -- cycle;
\draw[thick, ->] (1.95,0.5) arc (0:250:0.3);
\node (2) at (2,-0.5) {$2$-simplex};
\end{tikzpicture}
\caption{\label{simplices}Basic oriented simplices for the analysis of (2$D$) surfaces.}
\end{subfigure}
\begin{subfigure}{.45\textwidth}
\centering
\begin{tikzpicture}
    \begin{axis}[
       view={30}{60},axis lines=none,
       ]
       \addplot3[mesh,gray,
       samples=10,
       domain=0:2*pi,y domain=0:2*pi,
       z buffer=sort]
       ({(2+cos(deg(x)))*cos(deg(y))}, 
        {(2+cos(deg(x)))*sin(deg(y))}, 
        {sin(deg(x))});
\pgfplotsinvokeforeach{0,...,8}{        
 \addplot3[samples=10,gray,domain=0:360]
        ({(2+cos(x))*cos(x+#1*40)},
         {(2+cos(x))*sin(x+#1*40)},
         {sin(x)}); }       
   \end{axis}
\end{tikzpicture}
    \caption{\label{torus} A triangulated $T^2$ torus.}
\end{subfigure}
\caption{}
\end{figure}

More abstractly, we consider each $C_{\bullet}$ in the left chain complex to be the free abelian group generated by the finite set $K_{\bullet}$ such that $\left( C_{\bullet} , \partial^C_{\bullet} \right)$ is assumed to have a \textit{simplicial set} structure \cite{gelfand2013methods}. In order to avoid technicalities related to infinite dimensional Hilbert spaces later on, we restrict the right chain to have well defined group homomorphisms $\partial^G_{\bullet}$ between finitely generated abelian groups $G_\bullet$ (see the diagram above) satisfying the chain condition ($\partial^G_{n-1} \circ \partial^G_n = 0$), but otherwise no further structure is required at this level.

\subsection{\label{subsec:gaugeconfandrep}Gauge configurations and Gauge representations}

We call the sequence of morphisms $\omega_n : C_n\rightarrow G_{n-p}$ a $p$-map, where $p \in \mathbb{Z}$. In other words, there is a local assignment $ K_n \ni x \mapsto \omega_n \left( x \right) \in G_n$ which is performed by the simplicial pushforward $ x_* \omega := \omega \circ x = \omega_n (x)$ for all $x \in K_n$ . At the same time, the responsible for the localized compact support in the triangulated space $C\left( X \right)$ is the simplicial pullback  $x^* \left( y \right) = y_* x^* = \left\langle x | y \right\rangle_C := \delta \left(x,y\right)$ valid for all $x,y \in K$, where $\delta\left(\cdot,\cdot\right)$ is the Kronecker delta and the notation $\left\langle \cdot | \cdot \right\rangle_C$ signifies that this is indeed an inner product in the $\left( C_{\bullet} , \partial^C_{\bullet} \right)$ chain. Hence, we can decompose any $p$-map $\omega$ as the formal sum:  
\begin{align}\label{pmapdecomp}
\omega = \sum_{n , x \in K_n} \left(x_* \omega \right) \otimes x^* \quad ,
\end{align}
so that, for any chain $a = \sum_{n, y \in K_n} a_y y $ with $a_y \in \mathbb{Z}$ , we obtain the group ring element $\omega \left( a \right) = \sum_{n, y \in K_n} a_y \left(y_* \omega \right)$ , as expected. Writing $\text{Hom}(C_n,G_{n-p})$ for the set of all maps $\omega_n : C_n \rightarrow G_{n-p}$, we define: 
\begin{align}\label{homp}
\text{hom}(C,G)^p := \bigoplus_n \text{Hom}(C_n,G_{n-p}) \quad ,
\end{align}
which stans for the set of all of $p$-maps by construction. The latter becomes an abelian group under the binary operation:
\begin{align}\label{morphalg}
\left( \alpha + \beta \right)_n := \alpha_n + \beta_n \quad \text{with trivial identity morphism} \quad 0_n: C_n \rightarrow G_{n-p} \quad \forall  \; n \;\;\;\; ,
\end{align}
which turns this structure into a vector space (assuming trivial scalar multiplication) and provides a basis for $\text{hom}(C,G)^p$. It is convenient to think of the elements of this set as being represented by the diagram below: 
\begin{figure}[hb]
\begin{subfigure}{.49\textwidth}
  \centering
 \includegraphics[width=1\linewidth]{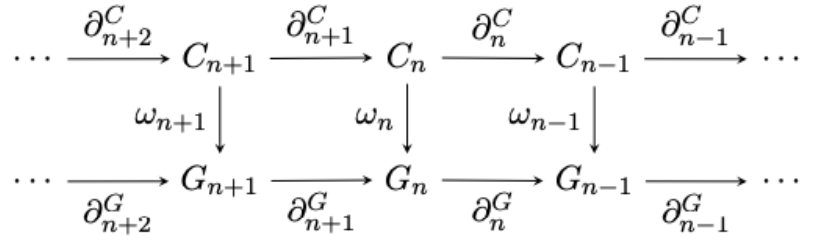}  
\caption{\label{fig:ChainMaps1} $\omega \in \text{hom} (C,G)^0$}
\end{subfigure}
\begin{subfigure}{.49\textwidth}
  \centering
  \includegraphics[width=1\linewidth]{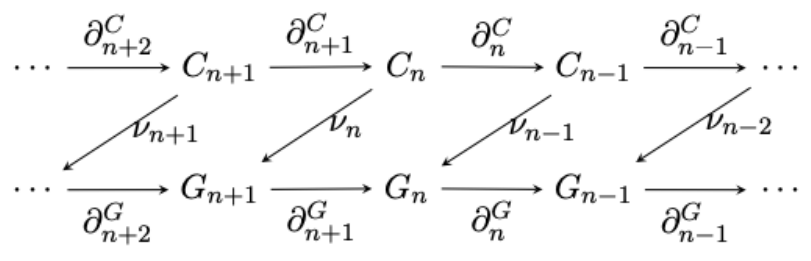}  
    \caption{\label{fig:ChainMaps2} $\nu \in \text{hom}(C,G)^{-1}$}
\end{subfigure}
\caption{\label{ChainMaps}Associate diagram for $p$-configurations}
\end{figure}

We follow the convention that for $p>0$ the lower chain $\left( G_{\bullet}, \partial^G_{\bullet} \right)$ is shifted $p$ steps to the right relative to the upper chain $\left( C_{\bullet}, \partial^C_{\bullet} \right)$ and similarly to the left for $p<0$.

In a similar rational, we call the sequence of morphisms $\hat{\rho}_n : C_n\rightarrow \hat{G}_{n-p}$ a $p$-map representation, where $\hat{G}_{n-p}$ is the set of irreducible representations of the elements of $G_{n-p}$ and $p \in \mathbb{Z}$. We decompose these as direct sums of compactly supported elements over $x \in K_n$ for all $n$:
\begin{align}\label{dualpmapdecomp}
\hat{\rho} = \sum_{n , x \in K_n} \left( x_* \hat{\rho} \right) \otimes x_* \quad \text{with} \quad  x_*\hat{\rho} := \hat{\rho} \circ x = \hat{\rho}_{n} \left( x \right) \in \hat{G}_{n-p} \quad \forall \;\; x \in K_n \quad ,
\end{align}
and we denote the set of all $p$-map representations as:
\begin{align}\label{cohomp}
\text{hom}(C,G)_p := \bigoplus_n \text{Hom}(C_n, \hat{G}_{n-p}) \quad .    
\end{align}

\begin{prop}[dual spaces isomorphism] \label{homiso} The spaces $\text{hom}(C,G)^p$ and $\text{hom}(C,G)_p$ satisfy the isomorphic relation $\text{hom} (C,G)^p \cong \text{hom} (C,G)_p$, $\forall p$.
\end{prop}
\begin{proof}
Since $G_{\bullet} \cong \hat{G}_{\bullet}$ for any abelian group, from (\ref{homp}) and (\ref{cohomp})  and the finiteness of these sets, we have $\left| \text{hom} (C,G)^p \right| = \left| \text{hom} (C,G)_p \right|$, from which the result follows.
\end{proof}

Analogously as with the configurations, the representation space also becomes a vector space by means of the binary operation:
\begin{align}\label{morphcoalg}
\left( \hat{\alpha} + \hat{\beta} \right)_n := \hat{\alpha}_n + \hat{\beta}_n \quad \text{with trivial representation} \quad \hat{0}_n: \omega_n \rightarrow \text{id}_{\text{GL}\left( V \right)_n} \quad \forall \; n \;\;\;\;.
\end{align}
Above we have denoted the identity element in the linear space $\text{GL}\left( V \right)_n$ as $\text{id}_{\text{GL}\left( V \right)_n}$. Thus, for any $\hat{\rho} \in \text{hom}(C,G)_p$ and $\omega \in \text{hom} (C,G)^p$, the expression:
\begin{align}
\nonumber \hat{\rho} \left( \omega \right) & = \sum_{x \in K} \left( x_* \hat{\rho}\right) \otimes x_* \circ \sum_{y \in K} \left( y_* \omega \right) \otimes y^* \\
\nonumber & = \sum_{x , y \in K} x_* \hat{\rho} \left( y_* \omega \right) \otimes \underbrace{x_* y^*}_{\delta \left( x , y \right)} = \sum_{x \in K} x_* \hat{\rho} \left( x_* \omega \right) \otimes 1 \\
\label{rho1} & \cong \sum_{x \in K} x_* \hat{\rho} \left( x_* \omega \right) := \sum_{x \in K} \hat{\rho} \left( \omega \right)_x \quad ,
\end{align} 
is well defined, where in the last line we have implicitly defined $\hat{\rho} \left( \omega \right)_x$ which stands as a shorthand of the representation $x_*\hat{\rho} \in \hat{G}_n$ of the element $x_* \omega \in G_n$ localized at $x \in K_n$.

Notice that the expression (\ref{rho1}) can be understood as the block diagonal matrix:
\begin{align}\label{rho2}
\hat{\rho} \left( \omega \right) = \left[ \begin{array}{ccccc}
    \hat{\rho} \left( \omega \right)_{x_1} & 0 & 0 & \cdots & 0 \\
    0 & \hat{\rho} \left( \omega \right)_{x_2} & 0 & \cdots &  0 \\
    0 & 0 & \hat{\rho} \left( \omega \right)_{x_2} & & \vdots \\
    \vdots & \vdots & & \ddots & 0 \\
    0 & 0 & \cdots & 0 & \hat{\rho} \left( \omega \right)_{x_{|K|}}
\end{array} 
\right] \quad \forall \; x_i \in K \quad .
\end{align}
The latter puts us in good foot to consider trace operators and consequently defining characters of these quantities, which we will explore below. In passing we can see that the expression (\ref{rho2}) is also a group homomorphism, satisfying:
\begin{align*}
\hat{\rho} \left( \omega_1 + \omega_2 \right) = \hat{\rho} \left( \omega_1 \right) + \hat{\rho} \left( \omega_2 \right) \quad \quad \text{and} \quad \quad \left[ \hat{\rho}_1 + \hat{\rho}_2 \right] \left( \omega \right) = \hat{\rho}_1 \left( \omega \right) + \hat{\rho}_2 \left( \omega \right) \quad ,  
\end{align*}
for all $\omega , \omega_1 , \omega_2 \in \text{hom} (C,G)^p$ and $\hat{\rho} , \hat{\rho}_1 , \hat{\rho}_2 \in \text{hom} (C,G)_p$. The previous discussion motivates a one-to-one correspondence between $p$-gauge configurations and $p$-gauge representations and the construction of some Hilbert space, so that quantum models can be defined on it. This will be done in the following subsection.

\subsubsection{\label{subsubsec:cohomology} Cohomology for the $p$-maps.}

In order to make explicit the topological features of these models a bit more of structure will be needed. The morphism $d^p:\text{hom} (C,G)^p \rightarrow \text{hom}(C,G)^{p+1}$ \cite{Brown} defined by:
\begin{align}\label{def:InternalHom}
\left( d^p \omega \right)_n = \omega_{n-1} \circ \partial^C_n - (-1)^p \partial^G_{n-p} \circ \omega_n 
\end{align}
turns $\left( \text{hom}(C,G)^\bullet  , d^\bullet \right)$ into a cochain complex:
\begin{align}\label{pseudochain1}
\cdots \xleftarrow{} \text{hom}(C,G)^{p+2} \xleftarrow{d^{p+1}} \text{hom}(C,G)^{p+1} \xleftarrow{d^{p}} \text{hom}(C,G)^{p} \xleftarrow{} \cdots \quad ,
\end{align}
i.e. the maps satisfy the relation $d^{p+1} \circ d^p=0$.

\begin{defn}[chain homotopy] \label{chainhomotopic1} Two elements  $\alpha, \beta \in \text{hom}(C,G)^p$ are said to be \textit{chain homotopic} if $\alpha = \beta + d^{p-1} \gamma$ for some $\gamma \in \text{hom}^{p-1}$.
\end{defn}
This relation defines an equivalence class in each degree of the cochain that we will denote by $\sim^{p}$ (the superscript $p$ will be omitted if no confusion arises). Hence, in equation (\ref{chainhomotopic1}), $\alpha$ and $\beta$ are said to be\footnote{Technically, chain homotopy is defined as the following sequence. Let $f, g : K_{\bullet} \rightarrow L_{\bullet}$ be two morphisms of chain complexes. A \textit{chain homotopy} between $f$ and $g$ is a sequence of homomorphisms
$h_n : K_n \rightarrow L_{n+1}$ such that $f_n - g_n = \partial^L_{n+1} \circ h_n + h_{n-1} 
\circ \partial^K_{n}$ for all $n \in  \mathbb{Z}$. This notion has been generalized in the obvious way.}  $p$-\textit{homotopic}, denoted $\alpha \sim^p \beta$ , so that cohomology groups can be defined in the standard way:

\begin{defn}[$p$-cohmology groups]\label{pcohomology} The cohomology of $G$ with coefficients in $C$ are defined through the quotient set $H^p(C,G) := \ker \left( d^p \right) / \text{im} \left( d^{p-1} \right)$, referred to as $p$-cohomology groups.
\end{defn}

\subsection{\label{subsec:hilb} Construction of the "$p$-Hilbert" spaces}

In this subsection we introduce the usual Dirac notation for quantum states and discuss the properties of the Hilbert spaces constructed from the $p$-maps and $p$-maps representations in a heuristic way. We will use the fact that a Hilbert space can always be constructed by means of defining an appropriate $C^*$ structure. The overall procedure is sketched as follows: Let $A$ be the $C^*$-algebra $C^* \left( G \right)$ (See App. \ref{app_GNS}), defined to be the $C^*$-enveloping algebra of $L^1 \left( G \right)$. Moreover, since $G$ is discrete, the well-definiteness of the $L^1$ norm is then ensured by having continuous representations. Consequently, any $*$-homomorphism from $\mathbb{C} \left[ G \right]$ to some $C^*$-algebra of bounded operators on some Hilbert space, factors through $\mathbb{C} \left[ G \right] \hookrightarrow C_{\max }^{*}(G)$. Additionally, for a field $k$, the spaces $k \left[ G_1 \right] \oplus_k  k \left[ G_2 \right] \cong k \left[ G_1 \oplus G_2 \right]$, which is inductively well defined for finite sums.

\subsubsection{\label{subsubsec:conf}The $p$-configuration basis}

From the previous discussion, and the matrix representation (\ref{rho2}), we can construct:
\begin{align}\label{hilbstates}
\text{hom} \left(C , G \right)^p \; \ni \; \omega \quad \mapsto \quad \ket{\omega} := \frac{1}{\left| \text{hom} (C,G)^p \right|^{1/2}} \sum_{\hat{\pi} \in \text{hom}(C,G)_p} \hat{\pi}^{\dagger} \left( \omega \right)  \quad ,
\end{align}
which is independent of the gauge representation and where the notation $\dagger$  stands for the conjugate transpose (as it is usually used in physics). This process results in the effective construction of a basis of a vector space that inherits the group structure from the $p$-maps. It is then easy to show from (\ref{hilbstates}) that its associated dual element is given by the linear functional:
\begin{align}\label{cohilbstates}
\text{hom} \left(C , G \right)^p \; \ni \; \omega \quad \mapsto \quad \bra{\omega}:= \frac{1}{\left| \text{hom} (C,G)^p \right|^{1/2}} \sum_{\hat{\pi} \in \text{hom}(C,G)_p} \text{tr} \circ \hat{\pi} \left( \omega \right) \circ \delta_{\hat{\pi}} \quad ,
\end{align}
where $\text{tr}$ stands for the trace operator. The latter can be used to show that the $(p)$ inner product:
\begin{align}
\nonumber \left\langle \nu | \omega \right\rangle_p & := \frac{1}{\left| \text{hom} (C,G)^p \right|} \sum_{\hat{\pi} \in \text{hom}(C,G)_p} \text{tr} \circ \hat{\pi} \left( \nu \right) \circ \delta_{\hat{\pi}} \sum_{\hat{\rho} \in \text{hom}(C,G)_p} \hat{\rho}^{\dagger} \left( \omega \right) \\
\label{innerproduct1} & = \frac{1}{\left| \text{hom} (C,G)^p \right|} \sum_{\hat{\pi} \in \text{hom}(C,G)_p} \left\langle \hat{\pi} \left( \nu \right) | \hat{\pi} \left( \omega \right) \right\rangle_F \quad ,
\end{align}
is well defined, where in the last line we have recognized the Frobenius inner product $\left\langle \cdot | \cdot \right\rangle_F$ for linear representations. 
It follows that (\ref{innerproduct1}) defines a norm by means of the Frobenius norm $\left\| \cdot \right\|_F$ given by $\left\| \omega \right\|^2_p := \left\langle \omega | \omega \right\rangle_p = \frac{1}{\left| \text{hom} (C,G)^p \right|} \sum_{\hat{\pi} \in \text{hom}(C,G)_p} \left\| \hat{\pi} \left( \omega \right) \right\|^2_F$. 
Summing up, we have constructed a huge set of inner product vector spaces of the form $\left( \text{span}_{\omega \in \text{hom}(C,G)^p} \left\{ \ket{\omega} \right\} , \left\langle \cdot | \cdot \right\rangle_p \right)$ from which a complete space can always be obtained by means of the Cauchy completion procedure.

\begin{defn}[$p$-Hilbert spaces]\label{hilbspace}
The Hilbert space of these models is defined by completion with respect to the norm above:
\begin{align*}
\mathcal{H}^p := \overline{\text{span}_{\omega \in \text{hom}(C,G)^p} \left\{ \ket{\omega} \right\}} = \bigoplus_{n, x \in K_n} \mathcal{H}^p_{n,x} \cong \bigoplus_{n, x \in  K_n} \mathbb{C} \left[G_{n-p}\right]_x\quad .
\end{align*}
\end{defn}

Where the last expression comes from thinking of each compactly supported space $\mathcal{H}^p_{n,x}$ as a $C^*$ structure and follow the procedure sketched above. We summarized this as the injection $\mathcal{H}^p_{n,x} \hookrightarrow \mathbb{C} \left[ G_{n-p} \right]_x$ , where the target stands for the group algebra of the abelian group $G_{n-p}$ compactly supported over the element $x \in K_n$. Hence, the dimension of the $p$-Hilbert spaces is finite and again given by:
\begin{align}\label{dimHilbert}
\text{dim}(\mathcal{H}^p) = | \text{hom} (C,G)^p | = \prod_n \left| G_{n-p} \right|^{\left| K_n \right|} < \infty \quad .
\end{align}
As it is known for Hilbert spaces, the inner product defines an isomorphism between $\mathcal{H}^p$ and its dual ${\mathcal{H}^p}^*$. In other words, as a byproduct we have $\text{dim} \left( {\mathcal{H}^p}^* \right) = \text{dim} \left( \mathcal{H}^p \right)$ .

\subsubsection{\label{subsubsec:repspace} The $p$-representation basis.}

There is, however, a second possibility for a well defined inner product space within this structure. We start by recalling that the characters of an irreducible representation $r$ of an element $g$ of a group $G$ are defined by the bounded linear functional $\chi_r \left( g \right) := \text{tr} \left( r \left( g \right) \right) \in \mathbb{C}$ \cite{isaacs1994character,fulton,serre}. The latter allows, by means of the Riesz representation theorem, to interpret the irreducible characters as a realization of the inner product $\chi_r \left( g \right) := \left\langle g , r \right\rangle := \left\langle r | g \right\rangle_G \;\; \in \;\; \mathbb{C}$ or, equivalently, $\chi_r : G \rightarrow \mathbb{C}$, which is reminiscent of the \textit{momentum eigenstates} of usual quantum mechanics \cite{cohen2018lecture}. The notation $\left\langle \cdot | \cdot \right\rangle_G$ indicates that this is indeed an inner product in the $\left( G_{\bullet} , \partial^G_{\bullet} \right)$ chain. In order to accommodate this perspective to our case, we first naively define the quantity:
\begin{align}\label{characinner}
\chi^{(p)}_{\hat{\pi}} \left( \omega \right) := \text{tr} \left( \hat{\pi} \left( \omega \right) \right)  \;\; \in \;\; \mathbb{C}  \quad \quad \text{for any} \quad \omega \in \text{hom}(C,G)^p \quad , \quad \hat{\pi} \in \text{hom}(C,G)_p \quad ,
\end{align}
as the $p$-\textit{character} of the $p$-gauge configuration $\omega$ in the $p$-gauge representation $\hat{\pi}$ and show that its properties coincide with that of the definition above.
This defines a skew-symmetric non-degenerate bilinear form, such that $\chi^{(p)}_{\hat{\pi}} \left( \cdot  \right) : \text{hom}(C,G)^p  \rightarrow  \mathbb{C}$ and $\chi^{(p)}_{ \left( \cdot \right)} \left( \omega \right) : \text{hom}(C,G)_p \rightarrow \mathbb{C}$ are well defined as maps. (We will systematically omit the superscript $(p)$ if no confusion arises). Consider now the following correspondence depending only on the gauge representation:
\begin{align}\label{cocohilbstates}
\text{hom}(C,G)_p \; \ni \; \hat{\pi} \quad \mapsto \quad \ket{\hat{\pi}} := \frac{\hat{\pi}^{\dagger} \left( 0 \right)}{\left\| \hat{\pi} \left( 0 \right) \right\|^{\frac{1}{2}}_F} \quad , \quad \bra{\hat{\pi}} := \text{tr} \circ \frac{\hat{\pi} \left( 0 \right)}{\left\| \hat{\pi} \left( 0 \right) \right\|^{\frac{1}{2}}_F} \circ \delta_{\hat{\pi}} \quad .
\end{align} 
It follows that $\left\langle \hat{\rho} | \hat{\pi} \right\rangle_p= \frac{\chi_{\hat{\pi}} \left( 0 \right)}{\left\| \hat{\pi} \left( 0 \right) \right\|_F}$ if $ \hat{\rho} = \hat{\pi}$ and zero otherwise.
So, if we interprete $\left\| \hat{\pi} \left( 0 \right) \right\|_F = \text{tr} \circ \hat{\pi}\left( 0 \right) \hat{\pi}^{\dagger}  \left( 0 \right) = \text{tr} \, \hat{\pi} \left( 0 \right) = \chi_{\hat{\pi}} \left( 0 \right) = \left| \text{hom}(C,G)^p \right|$, in analogy with its group counterpart, we obtain an orthonormal basis. We can do this since the graded group structure of $\left( \partial^G_{\bullet} , G_{\bullet} \right)$ is preserved by the trace operator, as is evident from Eq. (\ref{rho2}). By means of (\ref{hilbstates}) and (\ref{cohilbstates}), any $\omega \in  \text{hom}(C,G)^p$ and $\hat{\rho} \in \text{hom}(C,G)_p$, defines:
\begin{align}\label{mixedinnerprod} 
\left\langle \omega | \hat{\rho} \right\rangle_p 
= \frac{\text{tr} \, \hat{\rho} \left( \omega \right)}{\left\| \hat{\rho} \left( 0 \right) \right\|^{\frac{1}{2}}_F} := \frac{\chi_{\hat{\rho}} \left( \omega \right)}{\left| \text{hom}(C,G)^p \right|^{\frac{1}{2}}} \quad \text{and} \quad \left\langle \hat{\rho} | \omega \right\rangle_p  &
= \frac{\text{tr} \, \hat{\rho} \left( - \omega \right)}{\left\| \hat{\rho} \left( 0 \right) \right\|^{\frac{1}{2}}_F}  := \frac{\bar{\chi}_{\hat{\rho}} \left( \omega \right)}{\left| \text{hom}(C,G)^p \right|^{\frac{1}{2}}} \quad .
\end{align}
Furthermore, it is also straightforward to calculate:
\begin{align}\label{repinnerprod}
\sum_{\omega \in \text{hom}(C,G)^p} \frac{\bar{\chi}_{\hat{\pi}} \left( \omega \right) \chi_{\hat{\rho}} \left( \omega \right)}{\left| \text{hom}(C,G)^p \right|} = \sum_{\omega \in \text{hom}(C,G)^p} \frac{ \bar{\chi}_{\hat{\pi} - \hat{\rho}} \left( \omega \right)}{\left| \text{hom}(C,G)^p \right|} = \delta \left( \hat{\pi} , \hat{\rho} \right) = \left\langle \hat{\pi} | \hat{\rho} \right\rangle  \quad .
\end{align}
This is, the $p$-characters satisfy the firts Schur orthogonality relation \cite{isaacs1994character} for the charecters of irreducible representations, while the latter also serves as an alternative definition for the inner product in the \textit{gauge representation basis}. By abuse of terminology and, again in analogy with the group case, we indulge into call these $p$-gauge representations \textit{irreducible}, since the same irreducibility criterion is found in this context, i.e. $\left\| \hat{\pi} \right\|^2_p = 1$.

At this point, and similarly as before, we upgrade the inner product vector spaces $\left( \text{span}_{\hat{\rho} \in \text{hom}(C,G)_p} , \left\langle \cdot | \cdot  \right\rangle_p \right)$ (with the inner product defined as in (\ref{repinnerprod})) to a Hilbert space by a Cauchy completion procedure, i.e. $\hat{\mathcal{H}}^p := \overline{\text{span}_{\hat{\rho}} \left\{ \ket{\hat{\rho}} \right\}}$ for all $\hat{\rho} \in  \text{hom}(C,G)_p$. This implies:
\begin{align}\label{dimHilbertdual}
\text{dim}(\hat{\mathcal{H}}^p) = \left| \text{hom} (C,G)_p \right| = \left| \text{hom} (C,\hat{G})^p \right| = \prod_n \left| \hat{G}_{n-p} \right|^{\left| K_n \right|} = \text{dim}(\mathcal{H}^p)  \quad ,
\end{align}
where in the last step we have used the fact that $\hat{G}_{n-p} \cong G_{n-p}$ for abelian groups. Since both spaces are finite dimensional it follows that $\hat{\mathcal{H}}^p \cong \mathcal{H}^p$. As a consequence, the bases $\left\{  \ket{\hat{\rho}} \right\}_{\hat{\rho} \in \text{hom}(C,G)_p}$ and $\left\{ \ket{\omega} \right\}_{\omega \in \text{hom}(C,G)^p}$ are complete and orthonormal (the second by repeating the previous argument and using the second Schur orthogonality relation), allowing to write any state as in the following decompositions:
\begin{align}
\label{completerep} \mathcal{H}^p \; \ni \; &\ket{\Psi} = \sum_{\hat{\rho} \in \text{hom}(C,G)_p} \left\langle \hat{\rho} | \Psi \right\rangle \ket{\hat{\rho}}  \quad \quad \text{with} \quad \quad \mathbb{1}_{\mathcal{H}^p} = \sum_{\hat{\rho} \in \text{hom}(C,G)_p} \ket{\hat{\rho}} \otimes \bra{\hat{\rho}} \quad \quad  \text{and} \\
\label{completeconf} \mathcal{H}^p \; \ni \; &\ket{\Phi} = \sum_{\omega \in \text{hom}(C,G)^p} \left\langle \omega | \Phi \right\rangle \ket{\omega} \quad \quad \text{with} \quad \quad \mathbb{1}_{\mathcal{H}^p} = \sum_{\omega \in \text{hom}(C,G)^p} \ket{\omega} \otimes \bra{\omega} \quad .
\end{align}
where the second expression in each line is a resolution of the identity. 

\subsection{\label{subsec:homology} Induced Homology for $p$-gauge representations.}

For the rest of the paper we will write $\text{hom}(C,G)^p \mapsto \text{hom}^p$ and $\text{hom}(C,G)_p \mapsto \text{hom}_p$ to simplify the notation. Given the expressions (\ref{mixedinnerprod}), interpreted as inner product between the gauge configuration states and gauge representation states, we define the \textit{dual} operator $\hat{T}$ of $T \in \text{End} \left( \mathcal{H}^p \right)$ to be the unique adjoint operator with respect to this inner product. In other words, the operator $\hat{T}$ satisfies the relation $\left\langle \hat{\nu} | T | \omega  \right\rangle = \left\langle \hat{\nu} | T^* \omega  \right\rangle = \left\langle T_* \hat{\nu} | \omega \right\rangle := \left\langle \hat{T} \hat{\nu} | \omega \right\rangle$ for all $\omega \in \text{hom}^p$ and $\hat{\nu} \in \text{hom}_p$. Of all of the morphisms acting over $\text{hom}^p$, we are interested in the dual morphism of $d^p: \text{hom}^p \rightarrow \text{hom}^{p+1}$ , this is:
\begin{align}\label{compd_p}
\left\langle \hat{\nu} | d^p \omega \right\rangle = \left\langle \hat{d}^p \hat{\nu} | \omega \right\rangle := \left\langle d_{p+1} \hat{\nu} | \omega \right\rangle \quad ,
\end{align}
for all $\omega \in \text{hom}^p$ and $\hat{\nu} \in \text{hom}_{p+1}$ . Remark that in the last term we have implicitly defined $\hat{d}^p := d_{p+1}$ for our later convenience. The latter discussion leads to:

\begin{prop}[Compatibility for $d_p$ and $d_{p+1}$ maps]
In order for this map to be compatible, the diagram of FIG. \ref{commutdiagram}, defining a coproduct structure must be a commutative one.
\begin{figure}[!ht]
\begin{tikzcd}
\text{hom}_{p+1} \times \text{hom}^p \arrow{dr}{\varphi} \arrow[hook]{r}{d_{p+1} \otimes \sim} &  \text{im} \left( d_{p+1} \right) \times \text{im}\left( d^p \right)
\arrow[swap,dashrightarrow]{d}{\left\langle \cdot | \cdot \right\rangle_p} & \arrow[swap,hook]{l}{\sim \otimes d^p} \arrow[swap]{dl}{\varphi'} \text{hom}_p \times \text{hom}^{p-1} \\
& \sim 1 & 
\end{tikzcd}
\caption{\label{commutdiagram} Compatibility conditions for $d^p$ and $d_{p+1}$.}
\end{figure}
\end{prop}
Moreover,
\begin{prop}[]\label{isomorphismd} The sequence of maps $d^p$ and $d_{p+1}$ (defined above) satisfy the isomorphic condition $
\text{im} \left( d^p  \right) \cong \text{im} \left( d_{p+1} \right)$.
\end{prop}
\begin{proof}
The latter is immediate from non-degenerate inner product $\left\langle \cdot | \cdot \right\rangle_p$ and the universal property of the diagram \ref{commutdiagram}. 
\end{proof}

Additionally, under these conditions, for any $\omega \in \text{hom}(C,G)^{p-1}$, the following chain of equivalences must hold:
\begin{align}\label{generalboundary}
\left| \text{hom}(C,G)^{p+1} \right|^{1/2} = \left\langle \hat{\nu} | d^{p} \circ d^{p-1} \omega \right\rangle = \left\langle d_{p+1} \hat{\nu} | d^{p-1} \omega \right\rangle  = \left\langle d_{p} \circ d_{p+1} \hat{\nu} |  \omega \right\rangle \quad .
\end{align}\label{holonomymap}
The latter implies that the relation $d_{p} \circ d_{p+1} = 0$ must be satisfied for all $\hat{\nu} \in \text{hom}(C,G)_{p+2}$. In other words, the pair $\left( \text{hom}(C,G)_\bullet  , d_\bullet \right)$ becomes a chain complex:
\begin{align}\label{pseudochain2}
\cdots \xrightarrow{} \text{hom}(C,G)_{p+1} \xrightarrow{d_{p+1}} \text{hom}(C,G)_{p} \xrightarrow{d_{p}} \text{hom}(C,G)_{p-1} \xrightarrow{} \cdots \quad .
\end{align}
It follows that $d_p$ can be used to define:

\begin{defn}[chain comohotopy]\label{chainhomotopic2} 
Two elements  $\hat{\alpha}, \hat{\beta} \in \text{hom}(C,G)_p$ are said to be a \textit{co-chain homotopic} if they satisfy $\hat{\alpha} = \hat{\beta} + d_{p+1} \hat{\gamma}$ for some $\gamma \in \text{hom}(C,G)_{p+1}$.
\end{defn}

The latter also defines an equivalence class in each degree of the chain,  denoted by $\sim_{p}$. Hence, homology groups can also be defined in the canonical way

\begin{defn}[$p$-homology groups]\label{phomology}
The homology of $G$ with coefficients in $C$ are defined through the quotient sets $H_p(C,G) := \ker \left( d_{p+1} \right) / \text{im} \left( d_{p} \right)$ and are referred to as $p$-homology groups.
\end{defn}

In an explicit way, the morphism $d_p$ can also be understood in its components as a $p$-map representation by means of equation (\ref{dualpmapdecomp}). So, the inner product above can be written as $
\left\langle\hat{\rho} | \omega \right\rangle \propto \sum_{x,y \in K} \left\langle x_* \rho | y_* \omega \right\rangle_G \otimes \left\langle x | y \right\rangle_C$, which is the direct product of two inner products. Thus, an analogous rationale to that of the case of the $p$-maps, leads to the factoring:
\begin{align}\label{def:InternalcoHom}
\left( d_{p+1} f \right)_n = \hat{\rho}_{n} \circ \hat{\partial}^C_n - (-1)^p \hat{\partial}^G_{n-p} \circ \hat{\rho}_{n-1} 
\end{align}
where now $\hat{\partial}^C_n$ is the adjoint with respect to the $\left\langle \cdot | \cdot \right\rangle_C$ inner product, and $\hat{\partial}^G_n$ is the adjoint with respect to the $\left\langle \cdot | \cdot \right\rangle_G$ inner product. This is, $\left\langle x | \partial^C_n y  \right\rangle = \left\langle \hat{\partial}^C_n x | y \right\rangle_C$ for all $x , y \in K_n$, satisfying $\hat{\partial}^C_n \circ \hat{\partial}^C_{n-1} = 0$, and $\left\langle a | \partial^G_n b  \right\rangle = \left\langle \hat{\partial}^G_n a | b \right\rangle_G$ for all $a , b \in G_n$ satisfying $\hat{\partial}^G_n \circ \hat{\partial}^G_{n-1} = 0$. 

\section{\label{sec:QDM} Formalism  for Abelian Higher Symmetry Models}

We understand a $(0-)$ \textit{gauge configuration} to be performed by an element $\omega \in \text{hom} (C,G)^0$. Similarly, a $(0-)$ \textit{gauge representation} is understood to be performed by an element $\hat{\rho} \in \text{hom} (C,G)_0$. In what follows, we omit the $p$-grade of the Hilbert spaces $\mathcal{H}$ since it is understood that we are working over $\mathcal{H}^0$. Having said this, we describe the action of the operators acting over the bases of this latter Hilbert space and, it is clear that general behavior is immediately extrapolated it by linearity. Thus:
\begin{defn}[Shift and clock operators]\label{def:PQoperators}
Given  $\alpha \in \text{hom}^0$ and $\hat{\beta} \in \text{hom}_0$, we define the \textit{shift} and \textit{clock} operators, respectively as $P^{\alpha} \ket{\omega} = \ket{P^{\alpha} \circ \omega} := \ket{\omega + \alpha}$ and $Q_{\hat{\beta}} \ket{\omega} = \ket{Q_{\hat{\beta}} \circ \omega} := \chi_{\hat{\beta}} \left( \omega \right) \ket{\omega}$ for all  $\omega \in \text{hom}^0$.
\end{defn}

It is evident that the \textit{clock} operator is diagonal in the configuration basis. Meanwhile, even if the \textit{shift} operator is not, it allows us to study these models in a \textit{Heisenberg picture} fashion by means of the expression $\ket{\omega} = P^{\omega} \ket{0}$, where $0 \in \text{hom}^0$ is the aforementioned \textit{trivial} map.

From (\ref{def:PQoperators}) we immediately obtain (so we omit the proofs):
\begin{prop}[algebra\footnote{The strict abelian quantum double algebra is recovered when considering the associated operators $P^{\omega}$ as above and $R^{\mu}:=\frac{1}{\left| \text{hom}_0 \right|} \sum_{\hat{\rho}} \bar{\chi}_{\hat{\rho}} \left( \mu \right) Q_{\hat{\rho}}$ for $\omega \, ,\, \mu \, \in \text{hom}^0$ and $\hat{\rho} \in \text{hom}_0$. See for instance \cite{Beigi}.} for the shift and clock operators]\label{PQalgebra} The operators  defined in (\ref{def:PQoperators}) satisfy the relations: 
\begin{align*}
P^{\alpha} P^{\alpha'} = P^{\alpha'} P^{\alpha} = P^{\alpha + \alpha'} \;\;\; , \;\;\; Q_{\hat{\beta}} Q_{\hat{\beta}'} = Q_{\hat{\beta}'} Q_{\hat{\beta}} = Q_{\hat{\beta} + \hat{\beta}'} \;\;\; , \;\;\;  Q_{\hat{\beta}} P^{\alpha} = \chi_{\hat{\beta}} \left( \alpha \right) P^{\alpha} Q_{\hat{\beta}} \quad ,
\end{align*}
for all $\alpha, \alpha'\in \text{hom}^0$ and $\hat{\beta}, \hat{\beta}'\in \text{hom}_0$. 
\end{prop}

Requiring the operators defined in (\ref{def:PQoperators}) to be unitary implies that they satisfy $\left( P^{\alpha} \right)^{\dagger} = P^{-\alpha} $ and $\left( Q_{\hat{\alpha}} \right)^{\dagger} = Q_{-\hat{\alpha}}$ . Finally, it is a simple exercise to check that on the representation basis these operators reverse their role. This is, $Q_{\hat{\beta}} \ket{\hat{\alpha}} = \ket{\hat{Q}_{\hat{\beta}} \circ \hat{\alpha} } = \ket{\hat{\alpha} + \hat{\beta}}$ and $P^{\omega} \ket{\hat{\alpha}} = \ket{\hat{P}^{\omega} \circ \hat{\alpha}} = \bar{\chi}_{\hat{\alpha}} \left( \omega \right) \ket{\hat{\alpha}}$ for all $\hat{\alpha} , \hat{\beta} \in \text{hom}_0$ and $\omega \in \text{hom}^0$. Thus, the \textit{Heisenberg picture} in this basis is given by the clock operator, i.e. $\ket{\hat{\nu}} = Q_{\hat{\nu}} \ket{\hat{0}}$, where $\hat{0} \in \text{hom}_0$. Accordingly, throughout the paper, we will move freely between the configuration and representation bases when convenient.

\subsection{\label{subsec:module} Module structure over the configuration and representation bases.}

In order to define the dynamical operators suitable for our models in the next subsection, we will need to briefly describe the elements of $\text{hom}^{-1}$ and $\text{hom}_1$ and the mechanism in which they induce equivalence classes over the already discussed natural bases. Let us start with a general element $\alpha \in \text{hom}^{-1}$, represented diagrammatically in FIG. \ref{fig:ChainMaps2}.  Given $\text{hom}^0$, this set can be viewed as a $\hom^{-1}$-module by the action of the map $d^{-1}:\text{hom}^{-1} \rightarrow \text{hom}^0$. By extension, the Hilbert space $\mathcal{H}$ inherits the $\sim^0$ equivalence relation by performing $P^{d^{-1}\gamma}\ket{\beta} = \ket{\beta + d^{-1}\gamma} := \ket{\alpha}$, $\gamma \in \text{hom}^{-1}$. If this is the case we refer to $\ket{\alpha}$ and $\ket{\beta}$ as being \textit{gauge equivalent} and, by abuse of notation, we write $\ket{\alpha} \sim^0 \ket{\beta}$. Analogously, the set $\text{hom}_0$ can be viewed as a $\text{hom}_1$-module by the action of the map $d_1: \text{hom}_1 \rightarrow \text{hom}_0$. Consequently, over the Hilbert space $\hat{\mathcal{H}}$ is induced the $\sim_0$ equivalence when performing the operation $Q_{d_1\hat{\gamma}} \ket{\hat{\beta}} = \ket{\hat{\beta} + d_1 \hat{\gamma}} := \ket{\hat{\alpha}}$, $\hat{\gamma} \in \text{hom}_1$. Here, we call $\ket{\hat{\alpha}}$ and $\ket{\hat{\beta}}$ as \textit{co-gauge equivalent}, written $\ket{\hat{\alpha}} \sim_0 \ket{\hat{\beta}}$, where a general element of $\text{hom}_1$ corresponds to the diagram of FIG. \ref{fig:CoChainMaps2}. 
\begin{figure}[ht]
\begin{subfigure}{.49\textwidth}
  \centering
 \includegraphics[width=1\linewidth]{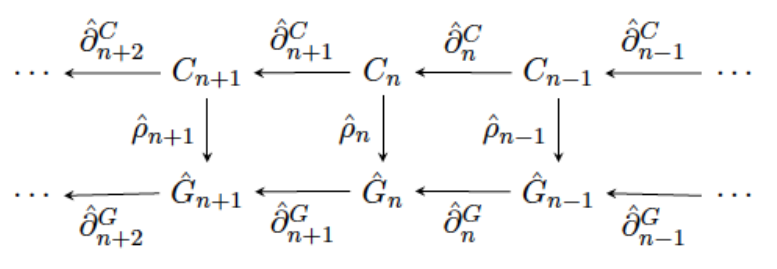}  
\caption{\label{fig:CoChainMaps} $\hat{\rho} \in \text{hom} (C,G)_0$}
\end{subfigure}
\begin{subfigure}{.49\textwidth}
  \centering
  \includegraphics[width=1\linewidth]{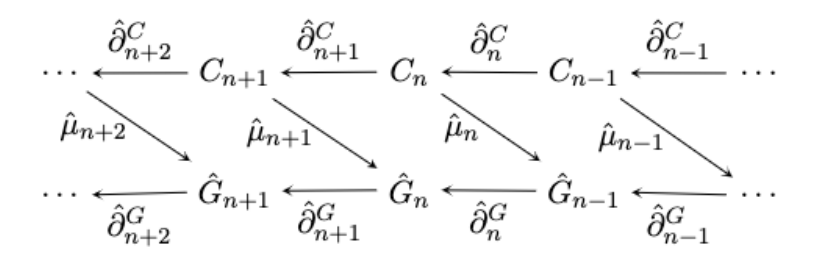}  
    \caption{\label{fig:CoChainMaps2} $\hat{\mu} \in \text{hom}(C,G)_1$}
\end{subfigure}
\caption{\label{CoCHainMaps}Associated diagram for $p$-representations.}
\end{figure}

It will be seen later on that it is this double module structure the one that is responsible for the topological order of these class of models. As of the usual topological order is concern, only the $\text{hom}^{-1}$ and $\text{hom}_1$ are sufficient to reproduce its effects. However, we leave open the possibility for studying other more general $\text{hom}^{-p}$ and $\text{hom}_p$-module structures ($p \neq 0,1$), since the formalism would allow it.

\subsection{\label{subsec:dynamics} Dynamics of the models}

We construct the operators: 
\begin{defn}[Fake-Gauge and Fake-Holonomy]
We define the \textit{fake-gauge} operator as the \textit{Fourier transform}:
\begin{align}
\label{Aop} A_{\hat{\rho}} := \mathcal{F} \left( P^{d^{-1}} \right) \left( \hat{\rho} \right) = \frac{1}{|\text{hom}^{-1}|}\sum_{\alpha\, \in \, \text{hom}^{-1}}\bar{\chi}_{\hat{\rho}}\left( \alpha \right)P^{d^{-1} \alpha} \quad & \text{for any} \quad \hat{\rho} \in \text{hom}_{-1} \quad ,
\end{align}
where $P^{d^{-1}}$ is the shift operator of equation (\ref{def:PQoperators}). Likewise, the \textit{fake-holonomy} operator is the one resulting from the \textit{inverse Fourier transform} of the clock operator:
\begin{align}
\label{Bop} B^{\omega} &:= \mathcal{F}^{-1} \left( Q_{d_1} \right) \left( \omega \right) = \frac{1}{|\text{hom}_{1}|}\sum_{\hat{\beta} \, \in \, \text{hom}_1} \bar{\chi}_{\hat{\beta}} \left( \omega \right) Q_{d_1 \hat{\beta}} \quad & \text{for any} \quad \omega \in \text{hom}^1 \quad ,
\end{align}
where $d_1:=\hat{d}^0$, which is dual relative to the \textit{character} inner product discussed in the last subsection.
\end{defn}

Explicitly, the operator $A_{\hat{\rho}}$ acting over an element $\ket{\omega} \in \mathcal{H}$ produces 
a weighted superposition of \textit{gauge equivalent} states. On the other hand, the fake-holonomy is diagonal over the elements $\ket{\nu} \in \mathcal{H}$.


\begin{prop}[Algebra of the Fake-Gauge and Fake-Holonomy operators]
For all $\hat{\rho},\hat{\rho}' \in \text{hom}_{-1}$ and $\omega, \omega'\in \text{hom}^1$, we can show that the operators (\ref{Aop}) and (\ref{Bop}) satisfy the algebra:
\begin{enumerate}
\item Pairwise commutation $\quad A_{\hat{\rho}} B^{\omega} = B^{\omega} A_{\hat{\rho}} \quad $,
 
\item Self-adjoint-ness $\quad A_{\hat{\rho}} = A^*_{\hat{\rho}} = A^{\dagger}_{\hat{\rho}}  \quad$ , $\quad B^{\omega} = \left( B^{\omega} \right)^* = \left( B^{\omega} \right)^{\dagger} \quad $, 

 
 
\item Orthogonality $\quad A_{\hat{\rho}} A_{\hat{\rho}'} = \delta \left(\hat{\rho},\hat{\rho}'\right) A_{\hat{\rho}'} \quad $ , $\quad B^{\omega} B^{\omega'} = \delta \left(\omega,\omega'\right) B^{\omega'} \quad $,


\item Completeness $\sum_{\hat{\rho}} A_{\hat{\rho}} = \mathbb{1}_{\mathcal{H}} \;\; , \;\; \sum_{\omega} B_{\omega} = \mathbb{1}_{\mathcal{H}}$ , where $\mathbb{1}_{\mathcal{H}}$ is the identity operator in ${\mathcal{H}}$ ,

 
In other words, this pair of operators are \textit{commuting orthogonal projectors}. 
\end{enumerate}
\end{prop}

The proof of this Proposition is an exercise in the orthogonal properties of characters and can be seen in the original paper \cite{dealmeida2017topological}. We emphasize that these operators act over compactly supported regions of the embedded space (in and around $x \in K$) by means of the decompositions (\ref{pmapdecomp})  and (\ref{dualpmapdecomp}), respectively. Following this reasoning and anticipating its connection with quantum codes, we can calculate:
\begin{align}\label{logAandB}
 \log_2 \left( A_{\hat{\rho}} \right) 
 = - \sum_{n, x \in K_n} \left( \mathbb{1}_{\mathcal{H}} - A_{\left(x_*\hat{\rho}\right)\otimes x_*} \right)  \;\; \text{and} \;\; \log_2 \left( B^{\omega} \right) 
 = - \sum_{n, x \in K_n} \left( \mathbb{1}_{\mathcal{H}} - B^{\left(x_*\omega\right)\otimes x^*} \right) \;\; ,
\end{align}
where we have expanded the natural logarithms using the Mercator series and we have used the fact that the localized terms $A_{\left( x_* \hat{\rho} \right) \otimes x_*}$ and $B^{\left( x_* \omega \right) \otimes x^*}$ for all $x \in K$ are projectors. Of particular interest to us is the trivial gauge operator $A_{\hat{0}}$ and trivial holonomy operator $B^{0}$. In order to understand why, let us first consider the state $\ket{\alpha} \in \mathcal{H}$ in the configuration basis. Here, the trivial gauge and holonomy operators of equations (\ref{Aop}) and (\ref{Bop}) act according to: 
\begin{align*}
A_{\hat{0}}\ket{\alpha} &= \dfrac{\sum_{\beta} P^{d^{-1} \beta}}{|\text{hom}^{-1}|} \ket{\alpha} = \dfrac{\sum_{\beta} \ket{\alpha + d^{-1} \beta}}{|\text{hom}^{-1}|} =\dfrac{1}{|\text{hom}^{-1}|}\sum_{\gamma \sim^0 \alpha} \ket{\gamma}  \;\;\;\;\;\;\;\;\; , \;\;\; \gamma \in \text{hom}^{-1} \quad , \\
B^{0}\ket{\alpha} &= \dfrac{\sum_{\hat{\gamma}} Q_{d_1 \hat{\gamma}}}{|\text{hom}_{1}|}\ket{\alpha} = \dfrac{\sum_{\hat{\gamma}} \chi_{\hat{\gamma}} \left( d^0 \alpha \right)}{{|\text{hom}_{1}|}} \ket{\alpha} = \delta \left( 0,d^0\alpha \right) \ket{\alpha} \quad \quad \quad \quad , \quad  \hat{\gamma} \in \text{hom}_{1} \quad . 
\end{align*}
On one hand, it can be seen that the second relation measures only the states in which $\alpha \in \ker d^0$, in other words with trivial (fake) holonomies or, equivalently, with flat connections. On the other hand, the first relation transforms the single state $\ket{\alpha}$ into a homogeneous superposition of states $\ket{\gamma}$ for all $\gamma \sim^0 \alpha \in \text{hom}^0$. Using the representation basis, for $\ket{\hat{\alpha}} \in \hat{\mathcal{H}}$ we obtain the reverse behavior $A_{\hat{0}} \ket{\hat{\alpha}} = \delta \left( \hat{0} ,d_0 \hat{\alpha} \right) \ket{\hat{\alpha}}$ and  $B^{0} \ket{\alpha} = \frac{1}{\left| \text{hom}_{1} \right|} \sum_{\hat{\gamma} \sim_0 \hat{\hat{\alpha}}} \ket{\hat{\gamma}}$ for $\hat{\gamma} \in \text{hom}_{1}$. In other words, on this basis $A_{\hat{0}}$ is measuring trivial co-gauge holonomies, and $B^0$ is producing a homogeneous superposition of states $\ket{\hat{\gamma}}$ for all $\hat{\gamma} \sim_0 \hat{\alpha} \in \text{hom}_0$. We use these to define:

\begin{defn}[Hamiltonian (\`{a} la Kitaev)] We define the Hamiltonian operator $H \in \text{End} \left( \mathcal{H} \right)$ as\footnote{The constant $\ln \left( 2 \right)$ is usually scaled to $1$ in most references, however we think this derivation is cleaner.}:
\begin{align}\label{def:Hamiltonian}
H := -\ln \left( A_{\hat{0}} B^{0} \right) = - \ln \left( 2 \right) \left[ \log_2 \left( A_{\hat0} \right) + \log_2 \left( B^{0}   \right) \right] \quad \quad .
\end{align}
\end{defn}
The minus sign is justified such that the Hamiltonian has a null eigenvalue whenever there is compactly supported gauge invariancy (enforced by $A_{\left(x_*\hat{0}\right) \otimes x_*}$)  and compactly supported trivial (fake) holonomy (enforced by $B_{\left(x_*0 \right) \otimes x^*}$). The more familiar expression is obtained when using equations (\ref{logAandB}). The projector nature of these compactly supported localized terms allows us to invert this expression and write:
\begin{align}\label{GSproj1}
 e^{-H} = A_{\hat{0}} B^{0} = B^{0} A_{\hat{0}} = \prod_{n,x \in K_n} A_{\left( x_* \hat{0} \right)\otimes x_*} \prod_{n,y \in K_n} B^{\left( y_* 0 \right)\otimes y^*} \quad .  
\end{align}

Denoting the algebra of observables as the closed set $\mathcal{O}=\mathcal{B}\left( \mathcal{H} \right)$. The generator of the dynamics \cite{1980applications,bratteli1987operator} is then the closure of the operator $\delta \left( O \right) := \left[ H , O    \right]$ valid for any $O \in \mathcal{O}$, where $H$ is the Hamiltonian operator defined in (\ref{def:Hamiltonian}). Thus, the time evolution is given by the expression $\mathcal{U}_t \left( O \right) := e^{i t \delta} \left( O \right) $. It easily follows that the dynamics of this class of systems, defined by the Hamiltonian (\ref{def:Hamiltonian}), performed by the Fake-Gauge and Fake-Holonomy operators, are time-independent.
The fact that these same operators perform time-independent dynamics over these models is completely desirable since the dynamical orbits defined by these flows remain unchanged.

\section{\label{sec:topdesc} Topological description of the Ground State Subspace}

This section mostly follows the treatment presented in \cite{dealmeida2017topological,Ibieta_Jimenez_2020}. We invite the reader to consult these references if further information on the formalism is needed. 
It follows from the discussion above that the spectrum of the Hamiltonian operator (\ref{def:Hamiltonian}) is bounded from below as $\sigma \left( H \right) \geq 0$ with equality when the condition:
\begin{align}\label{def:GS}
A_{\left(x_*\hat{0}\right)\otimes x_*} \ket{\Psi} = B^{\left(x_* 0\right)\otimes x^*} \ket{\Psi} = \ket{\Psi} \quad \quad , \quad \quad  \forall \;\; x \in K \quad , \quad \ket{\Psi} \in \mathcal{H}_0 \subset \mathcal{H} \quad ,
\end{align}
is satisfied. Where $\mathcal{H}_0$ is the \textit{ground state subspace} which can be defined by the closure of the span of all the states satisfying the condition above, i. e. $\mathcal{H}_0 = \overbrace{\left\{ \ket{\Psi} \right\}}$ such that $\ket{\Psi}$ satisfies (\ref{def:GS}).

\begin{defn}[Projector into the Ground state]
Let $\Pi_{\hat{0}}^{0} : \mathcal{H} \rightarrow \mathcal{H}_0$ be the operator:
\begin{align}\label{def:GSproj}
\Pi_{\hat{0}}^0 := e^{- H} \quad \quad \text{is a projector, such that} \quad \quad \Pi_{\hat{0}}^0 \ket{\Psi} =\begin{cases} 
\ket{\Psi} \;\; &, \;\;  \ket{\Psi} \in \mathcal{H}_0 \\
0 \;\; &, \;\; \text{otherwise} \\
\end{cases}
\quad ;
\end{align}
called the \textit{ground state projector}. 
\end{defn}

Notice that studying the eigenstates of this operator is equivalent to finding the states satisfying condition (\ref{def:GS}). Furthermore, the latter is precisely the operator defined in (\ref{GSproj1}) and trivially satisfies the commutation $\left[ \Pi_{\hat{0}}^0 , H \right]=0$, i.e. both operators have the same spectral resolutions. Given the general decomposition of an element of $\mathcal{H}$, we have $\Pi_{\hat{0}}^0 \ket{\Psi} = \sum_{\omega \in \text{hom}(C,G)^0} \Psi \left( \omega \right) \left(  \Pi_{\hat{0}}^0 \ket{\omega} \right)$. So, it is necessary and sufficient to study the behavior over the basis $\left\{ \omega \right\}_{\omega \in \text{hom}(C,G)^0}$  (or, equivalently, over the basis $\left\{ \ket{\hat{\rho}} \right\}_{\hat{\rho} \in \text{hom}(C,G)_0}$). Thus defined, the ground state degeneracy ($GSD$) can be computed by the expression: 
\begin{align}
\text{Tr} \left( \Pi_{\hat{0}}^0 \right) = \sum_i \lambda_i = GSD
\end{align}
where $\lambda_i$ are the eigenvalues of $\Pi_{\hat{0}}^0$ listed according to their algebraic multiplicity. In this case, $\lambda_i = 0,1$ out of being a projector and $\mu_A \left( 1 \right)$ (multiplicity of $\lambda_i = 1$) is exactly the $GSD$. At the same time, from (\ref{GSproj1}) we can calculate partition function associated to the embedded manifold $X$ to be: 
\begin{align}\label{partition}
 Z \left( X \right) = \text{Tr} \left( e^{-\beta H} \right) = \text{Tr} \left( \left( \Pi_{\hat{0}}^0 \right)^{\beta} \right) = \sum_i \lambda_i^{\beta} = {GSD}^{\beta} \quad \quad \text{with} \quad \quad  \beta = \frac{1}{\kappa_B T} \quad .   
\end{align}
It follows that, if the $GSD$ is topological, the partition function becomes topological as well, which is a general characteristic of TQFTs. Furthermore, the thermodynamic energy associated to these models is $\left\langle E \right\rangle = - \frac{\partial \ln Z \left( X \right)}{\partial \beta} = \ln \left( \frac{1}{GSD} \right)$, which, under the same assumptions, is also topological. However, notice that the usual thermodynamic entropy is always null $S = \kappa_B \left( \ln Z \left( X \right) + \beta \left\langle E \right\rangle\right) = 0$. Ergo, the necessity of defining new entropy like functions for the study of topological order. For a comprehensive approach on the context of these models see \cite{Ibieta_Jimenez_2020} where the topological entanglement entropy is defined and calculated.

\section{\label{ec:GSbases}Ground state subspace orthonormal bases}

In the following, we will make the best use of this dual structure so that we can study all possibilities at once. The situation is clear when one introduces some notation for later convenience. Let $\ket{\alpha_{\hat{\beta}}} := A_{\hat{\beta}} \ket{\alpha}$ (resp $\bra{\prescript{}{\hat{\beta}}{\alpha}}:=\bra{\alpha} A_{\hat{\beta}}$)  for $\alpha \in \text{hom}^0$ and $\hat{\beta} \in \text{hom}_{-1}$, so that in our models we can characterize any two \textit{gauge equivalent} states as satisfying $\ket{\alpha_{\hat{0}}} = \ket{\alpha'_{\hat{0}}}$ for $\alpha,\alpha' \in \text{hom}^0$, which can easily be verified to be valid if and only if $\alpha \sim^0 \alpha'$. Analogously, let $\ket{\hat{\alpha}^{\beta}} := B^{\beta} \ket{\hat{\alpha}}$ (resp.$\bra{\prescript{\beta}{}{\hat{\alpha}}} := \bra{\hat{\alpha}} B^{\beta}$) where $\hat{\alpha} \in \text{hom}_0$ and $\beta \in \text{hom}^{1}$. Consequently, two \textit{ cohomotopical} representations $\hat{\alpha} \sim_0 \hat{\beta} \in \text{hom}_0$ are \textit{fake co-gauge equivalent} if and only if $\ket{\hat{\alpha}^0} = \ket{\hat{\beta}^0}$. We first remark that the ground state subspace $\mathcal{H}_0$ is non-empty since it always has at least one element. Indeed, we can see this in the configuration basis when for the state $\ket{0_{\hat{0}}}$, which satisfies the condition $\Pi_{\hat{0}}^0 \ket{0_{\hat{0}}} = \ket{0_{\hat{0}}}$, i.e. $\ket{0_{\hat{0}}} \in \mathcal{H}_0$\footnote{The latter condition is sometimes referred in the condensed matter literature as \textit{frustration free} property \cite{wouters2021interrelations}. Concretely, given the Hamiltonian operator of equation (\ref{def:Hamiltonian}) we have the equivalent condition $H \Pi_{\hat{0}}^0 = 0$, fashionable in the quantum theory community \cite{Michalakis_2013}.} or, equivalently, in the representation basis, for the state $\ket{\hat{0}^{0}}$ that also satisfies the condition $\Pi_{\hat{0}}^0 \ket{\hat{0}^{0}} = \ket{\hat{0}^{0}}$, i.e. $\ket{\hat{0}^{0}} \in \mathcal{H}_0$. Since $\left[ \Pi^0_{\hat{0}} , P^{\omega} \right] = 0$ iff $\omega \in \ker \left( d^0 \right)$ (resp. $\left[ \Pi^0_{\hat{0}} , Q_{\hat{\rho}} \right] = 0$ iff $\hat{\rho} \in \ker \left( d_1 \right)$), the ground state subspace $\mathcal{H}_0$ turns out to be composed by states of the form: 
\begin{align}\label{groundstatesconf}
\ket{\omega_{\hat{0}}} = \left. P^{\omega} \ket{0_{\hat{0}}} \right|_{\omega \in \text{ker} \left( d^0 \right)}  \quad \quad ( \text{resp.} \quad \ket{\hat{\rho}^0} := Q_{\hat{\rho}} \left. \ket{\hat{0}^{0}} \right|_{\hat{\rho} \in \text{ker} \left( d_1 \right)} ) \quad \quad .  
\end{align} 
This turns out to be a complete characterization of $\mathcal{H}_0$, since every configuration is exhausted by the elements $\omega \in \text{hom}^0$ (resp. $\hat{\rho} \in \text{hom}_0$). Notice also that the state $\ket{0_{\hat{0}}} \in \mathcal{H}_0$ is intrinsic to all of the $\ket{\omega} \in \mathcal{H}_0$, this justifies $\ket{0_{\hat{0}}}$ being called the \textit{seed state for configurations} (resp. \textit{seed state for representations}), where its use can be seen in the literature, for example, in \cite{Padmanabhan:2021ywi}. The discussion so far has justified the claim that the elements of $\mathcal{H}_0$ can be classified by the equivalence classes of the relation $\sim^0$ (resp. $\sim_0$). This sets entail a natural basis for the ground state subspace so we are interested in the equivalence classes:
\begin{align}\label{0thcohomologyandhomology}
\ker \left( d^0 \right)/\sim^0 \; \cong \ker \left( d^0 \right)/\text{im} \left( d^{-1} \right) 
\quad \quad \text{and} \quad \quad \ker \left( d_1 \right)/ \sim_0 \; \cong \ker \left( d_1 \right)/\text{im} \left( d_0 \right) 
\quad \quad ,  
\end{align}
which follows from the fact that any two $\alpha, \beta \in \ker \left( d^0 \right)$ satisfying $\ket{\alpha_{\hat{0}}} = \ket{\beta_{\hat{0}}}$ (resp. $\hat{\alpha}, \hat{\beta} \in \text{ker} (d_0)$ satisfying $\ket{\hat{\alpha}^{0}} = \ket{\hat{\beta}^{0}}$), will belong to the same equivalence class if and only if $\alpha - \beta \in \text{im} \left( d^{-1} \right)$ (resp. ). More formally, the quotient space $\ker \left( d^0 \right)/\sim^0$ is characterized  by the commuting diagram shown in FIG.\ref{isomorphism1} (resp. $\ker \left( d_0 \right)/\sim_0$ in FIG. \ref{isomorphism2}), where the downward injection map is natural. The result then follows from the universal property. 
\begin{figure}[!ht]
\centering
\begin{subfigure}{.49\textwidth}
  \centering
\begin{tikzcd}
\text{ker} \left( d^{0} \right) \arrow[d,hook] \arrow{r}{d^{-1}} & \text{ker} \left( d^0 \right) / \text{im} \left( d^{-1} \right) \arrow{d}{\cong} \\
\text{ker} \left( d^{0} \right) / \sim^0 \arrow[ur,dashed] & H^0 \left( C, G \right)
\end{tikzcd}
\caption{ \label{isomorphism1} Commutative diagram for configurations. }
\end{subfigure}
\begin{subfigure}{.49\textwidth}
  \centering
\begin{tikzcd}
\text{ker} \left( d_1 \right) \arrow[d,hook] \arrow{r}{d_0} & \text{ker} \left( d_1 \right) / \text{im} \left( d_0 \right) \arrow{d}{\cong}  \\
\text{ker} \left( d_1 \right) / \sim_0 \arrow[ur,dashed] & H_0 \left( C, G \right)
\end{tikzcd}
\caption{ \label{isomorphism2} Commutative diagram for representations. }
\end{subfigure}
\caption{}
\end{figure}

This discussion leads to the following:
\begin{prop}[GSD]\label{thm:main} The dimension of the subspace $\mathcal{H}_0$ (GSD) is given by: 
\begin{enumerate}
\item  $\text{GSD} = |H^0(C,G)|$ (the main result of \cite{dealmeida2017topological}), or equivalently,
\item $\text{GSD}=\left| H_0 (C,G) \right|$.
\end{enumerate}
\end{prop}
It is clear that these expressions are topological in the strict mathematical sense, since they are defined by the specific cohomology group $H^0 (C,G)$ and homology group $H_0 (C,G)$. Then it is evident that a suitable basis for $\mathcal{H}_0$ can be labeled by the equivalence classes $\underline{\omega} \in H^0 (C,G)$, where $\omega$ is a representative element or, equivalently, by $\hat{\underline{\nu}} \in H_0 (C,G)$, where $\nu$ is a representative element. These are immediately orthonormal by construction. We can reconcile these two bases by profiting from the already obtained isomorphisms:

\begin{prop}[Compatibility between bases of the ground state subspace]\label{thm:mainhom} The bases $\left\{ \ket{\underline{\omega}} \right\}_{\underline{\omega} \in H^0}$ and $\left\{ \ket{\underline{\hat{\rho}}} \right\}_{\underline{\hat{\rho}} \in H_0}$ are complete and equivalent as bases for the ground state $\mathcal{H}_0$.
\end{prop}
\begin{proof}
Recall that, from (\ref{isomorphismd}) with $p=0$, we have shown that $\text{im} \, \left( d^{-1} \right) \cong \text{im} \, \left( d_0 \right)$. At the same time, from the first isomorphism theorem, we have $\text{im} \, \left( d^0 \right) \cong \text{hom}^{0} / \ker \left( d^0 \right)$ and $\text{im} \, \left( d_1 \right) \cong \text{hom}_0 / \ker \left( d_1 \right)$. Given that $\text{hom}^0 \cong \text{hom}_0$ (see (\ref{homiso}) for $p=0$), it follows that $\ker \left( d^0 \right) \cong \ker \left( d_1 \right)$ and, consequently $H^0 (C,G) \cong H_0 (C,G)$. Furthermore, we have:
\begin{align*}
\text{GSD} = |H_0 (C,G)| =|H^0(C,G)| \quad \quad .
\end{align*}
In other words, the gauge representation basis is also a complete basis for $\mathcal{H}_0$.
\end{proof}

 This dualization between homology and cohomology falls into a more general framework. In fact, given the definitions used in this paper, the same rationale leading to (\ref{thm:mainhom}), implies:
\begin{prop}[duality for $p$-homology and $p$-cohomology] \label{pduality} The mathematical structure of these  abelian systems imply the isomorphism $\hat{H}^p (C,G) \cong H^p (C,G) \cong H_p (C,G)$.
\end{prop}
\begin{proof} The latter follows from the previous results $
\text{im} \left( d^{p-1}\right) \cong \text{im} \left( d_p \right)$  and $\text{ker} \left( d^{p-1} \right) \cong \text{ker} \left( d_p \right)$ which is induced by the bilinear form mixed inner product.
\end{proof}
This duality can be thought of as a consequence of the deep symmetry in the structure of these models. Basically, this construction carries the property that the dual of a group and the group itself are again isomorphic when abelian \cite{serre} to the level of topology.

Even though the ground state is sufficiently characterized at this point, and its degeneracy can be analytically computed (modulo the group structure of the $\left( G_{\bullet} \partial^G_{\bullet}\right)$ chain, see App. \ref{app:dimanalysis}), these bases are not enough for classification purposes. In fact, from a purely algebraic perspective, we can state some immediate facts derived from the previous sections. Firstly, $\delta \left( P^{\underline{\omega}} \right) = \delta \left( Q_{\underline{\hat{\nu}}} \right) = \mathbb{1}_{\mathcal{H}}$ iff $U_t \left( P^{\underline{\omega}} \right) = U_t \left( Q_{\underline{\hat{\nu}}} \right) = \mathbb{1}_{\mathcal{H}}$  for all $\underline{\omega} \in H^0$ and $\underline{\hat{\nu}} \in H_0$. Hence, these operators also perform dynamical transformations over the ground state space, which can be recognized as the generators from a so called seed state. At the same time, it is easy to show that the following commutation relations are satisfied for the operators above:
\begin{align}\label{gencomm1}
\left[ \Pi_{\hat{0}}^0 , P^{\underline{\omega}} \right] = 0 \quad , \quad \left[ \Pi_{\hat{0}}^0 , Q_{\underline{\hat{\nu}}} \right] = 0 \quad , \quad \left[ \Pi_{\hat{0}}^0 , H \right]=0 \quad , \quad \left[ Q_{\underline{\hat{\nu}}} , P^{\underline{\omega}} \right]_g \propto \chi_{\underline{\hat{\nu}}} \left( \underline{\omega} \right) \quad ,
\end{align}
where the last expression is reminiscent of the \textit{momentum eigenstates} commutation, and where we have denoted $\left[ \cdot , \cdot \right]_g$ for the \textit{group commutator}. It follows then that the operators $P^{\underline{\omega}}$ and $Q_{\underline{\hat{\nu}}}$ \textit{do not conform} a complete set of commuting observables for the ground state subspace projected by $\Pi_{\hat{0}}^0$. 
The strategy will be to bypass this complication by interpreting this space in a more \textit{algebro-topological} way in the following section.

\subsection{\label{app:dimanalysis} On the Ground state Degeneracy Calculation (*)\footnote{This section can be omitted without compromising the reading of the rest of the paper.}}

As before and, for economical reasons, we simplify the notation once more and write $H^p(C,G) \mapsto H^p$ and $H_p(C,G) \mapsto H_p$ whenever no confusion arises. Before embarking on the classification problem, we briefly explore the ground state degeneracy calculation in an explicit dimensional way. We focus on its cohomological description since the use of the Universal Coefficient Theorem is immediate. However, as proven in the previous subsection, homology should equally adequate to study this description. To begin, it follows from (\ref{groundstatesconf}) that:
\begin{align}\label{Pgs}
\mathcal{H}_0 \; \ni \, \ket{\underline{\omega}_{\hat{0}}} = P_{\underline{\omega}} \ket{\underline{0}_{\hat{0}}} \quad \quad \text{where} \quad \quad  \underline{\omega} \in H^0 \quad .
\end{align}
and hence, the structure of of the basis is decoded in the operator $P_{\underline{\omega}}$ appearing in (\ref{Pgs}) which, when studied from a \textit{Heisenberg picture perspective}, can be understood as a generator of the configuration basis. On one hand, any ground state basis as can be decomposed as:
\begin{align}\label{0decomp}
H^0 \ni \underline{\omega} = \bigoplus_{n} \underline{\omega}_n \quad \quad \text{where} \quad \quad \underline{\omega}_n := \bigoplus_{x \in K_n} \underline{\left( x_* \omega \right)} x^* \quad .    
\end{align}
On the other hand, for the sake of definiteness, consider the (not necessarily natural) injection map $\iota: \prod_n H^n ( C, H_{n}(G)) \rightarrow H^0(C,G)$ defined as:
\begin{align}\label{omegadecompn}
\underline{\omega} = \iota \underline{\alpha} := \sum_n \pi_n \underline{\alpha} \quad \text{where} \quad \underline{\alpha} \in \prod_n H^n ( C, H_{n}(G)) \quad \text{with} \quad \pi_n \underline{\alpha} = \underline{\omega}_n \quad ,
\end{align}
where the latter is understood as a formal sum. A similar map is used in \cite{dealmeida2017topological} and it is proved in \cite{Brown} to be an isomorphism\footnote{As it is shown in \cite{Brown}, such isomorphism exists for all $\prod_n H^n ( C, H_{n-p}(G)) \rightarrow H^p(C,G)$ in any degree $p$.}. Hence, the ground state operator of equation (\ref{Pgs}) allows the decomposition $P_{\underline{\omega}} = P_{\bigoplus_n \underline{\omega}_n} \cong P_{\sum_n \pi_n \underline{\alpha}} = \prod_n P_{\pi_n \underline{\alpha}} \cong \prod_n P_{\underline{\omega}_n}$. In other words, $P_{\underline{\omega}}$ is performing many operations in different dimensions simultaneously. As a byproduct, the elements of the cohomology groups themselves define an $n$-dimensional decomposition: 
\begin{align*}
\ket{\underline{\omega}_{\hat{0}}} = P_{\underline{\omega}}\ket{0_{\hat{0}}} = \prod_n P_{\underline{\omega}_n} \ket{0_{\hat{0}}} = \bigoplus_n \ket{ \pi_n \underline{\alpha}_{\hat{0}}} \cong \bigotimes_n \ket{{\underline{\omega}_n}_{\hat{0}}} \quad , \quad \pi_n \underline{\alpha} \in H^n ( C, H_{n}(G)) \quad .
\end{align*} 
The same reasoning can be employed to understand $Q_{\underline{\hat{\rho}}}$ 
as  being the generator of the dual basis and, consequently, decompose any dual space basis as $\ket{\hat{\underline{\rho}}_0} \cong \bigotimes_n \ket{{{\hat{\underline{\rho}}}_n} \, {}_0}$ with $\hat{\underline{\rho}} = \widehat{{\iota \underline{\alpha}}}$ and $\underline{\alpha} \in \prod_n H^n ( C , H_{n}(G))$. At the same time, the inner structure of $H^n ( C, H_{n-p}(G)$ for all $n$ can also be understood directly from the universal coefficient theorem \cite{Hatcher}:
\begin{align}\label{UCT} 
H^n \left( C , H_n \left( G \right) \right) = \text{Hom} \left( H_n \left( C \right) ,  H_n \left( G \right) \right) \oplus \text{Ext}^1 \left( H_{n-1} \left( C \right) ,  H_n \left( G \right) \right) \quad ,
\end{align}
which makes explicit the fact that the ground state is completely independent of the particular triangulation $C(X)$ used.

Let us first tackle the $\left( C_{\bullet}, \partial^C_{\bullet} \right)$ chain. Most of the examples worked so far in the literature have considered free modules $H_{\bullet} \left( C \right)$ for some simplicial complex $C(X)$ \cite{dealmeida2017topological,Ibieta_Jimenez_2020}. This results in the trivialization of the $\text{Ext}^1 \left( H_{n-1} \left( C \right) , - \right)$ functor in (\ref{UCT}) for all $n$. The same result holds more generally for all $H_{\bullet} \left( C \right)$ projective (free $\Rightarrow$ projective) \cite{Weibel}. By means of the fundamental theorem of finitely generated abelian groups, any Homology group can be decomposed into:
\begin{align}\label{simphomology}
H_n \left( C \right) = F_n \oplus T_n \quad \text{with} \quad F_n \cong \mathbb{Z}^{\beta_n} \quad \text{and} \quad  T_n \cong \bigoplus_
{\left\{p\right\}_n} \mathbb{Z} / \left\{ p \right\}_n \mathbb{Z} \quad ,
\end{align}
where $F_n$ is the free part and $T_n$ is the torsion part of the $n$-th Homology group $H_n \left( C \right)$. The $n$-th free part is isomorphic to $\mathbb{Z}^{\beta_n}$, where $\beta_n$ is the $n$-th dimensional Betti number. The $n$-th torsion part is isomorphic to a direct sum of cyclic groups of prime power group orders which we have collectively denoted by $\left\{p\right\}_n$.

Now, obtaining $H_n \left( G \right)$ from the $\left( G_{\bullet}, \partial^G_{\bullet} \right)$ chain is usually non-straightforward. In general, given a group $G$, the homology groups $H_n \left( G \right)$ are defined as $H_n \left( G \right) = H_n \left( \mathbb{Z} \otimes_{\mathbb{Z}G} F_* \right)$ for $n \in \mathbb{N}$, where $\mathbb{Z}G$ is the free $\mathbb{Z}$-module generated by the elements of $G$ and $F_*$ is any free (or projective) resolution for $G$ \cite{MR672956}. However, the problem of constructing a free resolution for general abelian groups is highly non-trivial. Moreover, it is known to be carried out directly in the cases where the resolution is small (or minimal) so that several computation algorithms have been devised \cite{romero2011computing,RUBIO2002389}.

In our case, as typical examples of the type of chains $\left( G_{\bullet}, \partial^G_{\bullet} \right)$, a cyclic group $G$ of order $m$ are usually considered. Several structures can be embedded in this basic framework. For instance, the usual alternative  is considering $G$ to be the graded group $G=\bigsqcup_{p} G_{p}$, such that each $G_p \cong N_p \rtimes G_{p-1}$ (semi direct product), where $N_p$ is a normal subgroup and the action of $G_{p-1}$ is performed via the inverse group morphisms $\left(\partial^G_p\right)^{-1}$. As it is known, this decomposition is not unique and obtaining its resolution is usually cumbersome. Instead, let us consider a known resolution of the cyclic group $G$ of order $m$, which is characterized by its single generator $t$ so that $F_*$ is given by:
\begin{align*}
\cdots \xrightarrow{N} \mathbb{Z} G \xrightarrow{t-1} \mathbb{Z} G \xrightarrow{N} \mathbb{Z} G \xrightarrow{t-1} \mathbb{Z} G \rightarrow \mathbb{Z} \rightarrow 0 \quad ,
\end{align*}
which, for definiteness, we take to coincide with the chain $(G_{\bullet},\partial^G_{\bullet})$, where $N=1+t+t^2+t^3+\cdots+t^{n-1}$ is the normal element of $\mathbb{Z}G$. The latter produces the homology groups ($n>0$):
\begin{align}\label{grouphomology}
H_n \left( G \right) = \begin{cases}
\mathbb{Z} &, \quad \quad n = 0 \\
\mathbb{Z} / n \mathbb{Z} &, \quad \quad n \text{ odd}\\
0 &, \quad \quad n \text{ even}
\end{cases}
\end{align}

When inserting these back into equation (\ref{UCT}), we obtain the following expression in terms of the $n$-th group homology:
\begin{align}\label{Classification}
H^n \left( C , H_n \left( G \right) \right) \cong \begin{cases} 
\mathbb{Z}^{\beta_n} \bigoplus_{\left\{ p \right\}_n} \prescript{}{\left\{ p \right\}_n}{\mathbb{Z}} &, \quad \quad n = 0 \\
\mathbb{Z}_n^{\beta_n} \bigoplus_{\left\{ p \right\}_{n-1}} \mathbb{Z}_{\text{gcd}\left( m , \left\{ p \right\}_{n-1}\right)} \bigoplus_{\left\{ p \right\}_n} \prescript{}{\left\{ p \right\}_n}{\mathbb{Z}_n} &, \quad \quad n \text{ odd} \\
\mathbb{Z}^{\beta_n} \bigoplus_{\left\{ p \right\}_{n-1}} \mathbb{Z}_{\left\{ p \right\}_{n-1}} \bigoplus_{\left\{ p \right\}_n} \prescript{}{\left\{ p \right\}_n}{\mathbb{Z}} &, \quad \quad n \text{ even}
\end{cases} \quad .
\end{align}
where we have written $\mathbb{Z}_d \cong \mathbb{Z} / d \mathbb{Z}$ with $d \in \mathbb{N}$ and we have denoted as $\prescript{}{h}{G}$ the subgroup $\left\{ g  \in G : h g = 0 \right\}$ of elements of $G$ of order $h$ (or some divisor of $h$) and $gcd \left( a,b \right)$ stands for the greatest common divisor between $a$ and $b$.

\section{\label{sec:classification} Classification of the Ground State Subspace}

Coming back to the ground state subspace bases, we will first show that they are not only related by an endomorphism, but they are unitarily equivalent. Let us start by focusing on the representation space, since the configuration space case is analogous. Consider two elements $\ket{\underline{\hat{\alpha}}^{0}} = B^{0}\ket{\underline{\hat{\alpha}}}\, , \, \ket{\underline{\hat{\beta}}^{0}} = B^{0}\ket{\underline{\hat{\beta}}}$ of the cohomological basis for $\mathcal{H}_0$, where $\underline{\hat{\alpha}}, \underline{\hat{\beta}} \in H_0$. Taking its inner product, yields:
\begin{align*}
\left\langle \prescript{0}{}{\underline{\hat{\alpha}}} | \underline{\hat{\beta}}^{0} \right\rangle = \delta \left( \underline{\hat{\alpha}} , \underline{\hat{\beta}} \right) + \frac{1}{\left| \text{hom}_1 \right|} \sum_{\hat{\mu} \in \text{hom}_1} \delta\left( \underline{\hat{\alpha}}, d_{1} \hat{\mu} \right) = \delta \left( \underline{\hat{\alpha}} , \underline{\hat{\beta}} \right) \quad \quad ,
\end{align*}
where the term under sum is identically null since $\underline{\hat{\alpha}} \notin \text{im} \left( d_1 \right)$. Thus, this is indeed an orthonormal basis. The situation is completely analogous for the basis $\left\{ \ket{\underline{\alpha}_{\hat{0}}} \right\}_{\underline{\alpha} \in H^0} \,$. Consequently: 
\begin{prop}[\label{restrictionGS}Restriction to the ground state subspace]
Restricting the inner product to the ground state subspace $\mathcal{H}_0$ (denoted $\left\langle \cdot | \cdot \right\rangle_{\mathcal{H}_0}$), we have:
\begin{enumerate}
\item The dimension of the groundstate subspace to be given by $\left.\chi_{\underline{\hat{\pi}}} \left( \underline{0} \right)\right|_{\mathcal{H}_0} = \left| H^0 \right|$ , 

\item The string of relations for the mixed inner product:
\begin{align}\label{mixedGSS}
\left\langle  \prescript{}{\hat{0}}{\underline{\omega}} | \underline{\hat{\nu}}^0 \right\rangle_{\mathcal{H}_0} = \bra{\underline{\omega}} A_{\hat{0}} B^{0} \ket{\underline{\hat{\nu}}}_{\mathcal{H}_0} =  \bra{\underline{\omega}} B^{0} A_{\hat{0}} \ket{\underline{\hat{\nu}}}_{\mathcal{H}_0} = \left\langle \underline{\omega} | \underline{\hat{\nu}} \right\rangle_{\mathcal{H}_0} = \frac{\chi_{\underline{\hat{\nu}}} \left( \underline{\omega} \right)}{\left| H^0 \right|^{\frac{1}{2}}}
\end{align}
and 
\item the following resolutions of the identity (completeness relations):
\begin{align}
\nonumber & \Pi^{\underline{\omega}} := \ket{\underline{\omega}_{\hat{0}}} \otimes \bra{\prescript{}{\hat{0}}{\underline{\omega}}}  \quad \quad \text{such that} \quad \quad \sum_{\underline{\omega} \in H^0} \Pi^{\underline{\omega}} = \mathbb{1}_{\mathcal{H}_0} \quad \quad \text{and}  \\ 
\label{completeGSS} & \Pi_{\underline{\hat{\nu}}} := \ket{\underline{\hat{\nu}}^0} \otimes \bra{\prescript{0}{}{\underline{\hat{\nu}}}} \quad \quad \text{such that} \quad \quad \sum_{\underline{\hat{\nu}} \in H_0} \Pi_{\underline{\hat{\nu}}} = \mathbb{1}_{\mathcal{H}_0} \quad \quad ,
\end{align}
where $\Pi^{\underline{\omega}}$ and $\Pi_{\underline{\hat{\nu}}} $ are projectors into the subspaces $U^{\underline{\omega}} := \text{span} \ket{\alpha_{\hat{0}}}_{\alpha \sim^0 \underline{\omega} \in H^0}$ and $U_{\underline{\hat{\nu}}} :=\text{span} \ket{\hat{\beta}^{0}}|_{\hat{\beta} \sim_0 \underline{\hat{\nu}} \in H_0}$ of $\mathcal{H}_0$, respectively. 
\end{enumerate}
\end{prop}

The proofs are straightforward so we also omit them here. However, from (\ref{gencomm1}) it is clear that $\Pi^{\underline{\omega}}$ and $\Pi_{\underline{\hat{\nu}}}$ are still \textit{non} commuting projectors. Nevertheless, the completeness of the bases imply that the closed indexed sets $\left\{ \overline{U^{\underline{\omega}}} \right\}_{\underline{\omega} \in H^0}$ and $\left\{ \overline{U_{\underline{\hat{\nu}}}}\right\}_{\underline{\hat{\nu}} \in H_0}$ decompose the ground state subspace as $\mathcal{H}_0 \cong \bigoplus_{\underline{\omega} \in H^0} \overline{U^{\underline{\omega}}} \cong \bigoplus_{\underline{\hat{\nu}} \in H_0} \overline{U_{\underline{\hat{\nu}}}} $. This decomposition implies that the space $\mathcal{H}_0$ is at least connected to the class $\ket{\underline{0}_{\hat{0}}}$ or $\ket{\underline{\hat{0}}^0}$, respectively. This is not enough for our purposes, so we embark on studying some of its topological properties as to ensure that the classification space is well defined.

\subsection{\label{subsec:pathconnec}Unitary equivalence, path-connectedness and open covers}

Given that we known two complete bases, it is sufficient to study them and extend by linearity any results obtained. Let us first show that the bases of $\mathcal{H}_0$ transform unitarily.

\begin{prop}[Unitary equivalence] \label{Uequivalence}
The bases found in the previous section are \textit{unitarily equivalent}.
\end{prop}
\begin{proof}
Using the completeness relations (\ref{completeGSS}) and relations (\ref{mixedGSS}) above, we can show that the basis elements transform into each other as:
\begin{align}\label{transfunc} 
\ket{\underline{\omega}_{\hat{0}}} = \left[ U \right]^{\underline{\hat{\alpha}}}_{\underline{\omega}} \ket{\underline{\hat{\alpha}}^0} \quad \text{and} \quad \ket{\underline{\hat{\nu}}^0} = \left( \left[ U \right]^{\underline{\hat{\alpha}}}_{\underline{\omega}} \right)^{\dagger} \ket{\underline{\beta}_{\hat{0}}} \quad \quad \text{where} \quad \left[ U \right]^{\underline{\hat{\alpha}}}_{\underline{\omega}} := \frac{\bar{\chi}_{\underline{\hat{\alpha}}} \left( \omega \right)}{\left| H^0 \right|^{\frac{1}{2}}} \quad,
\end{align}
written in matrix notation. Notice then, that $\left[ U \right]^{\underline{\hat{\alpha}}}_{\underline{\omega}}$ is a unitary matrix, i.e.  
\begin{align*}
\left[ U \right]^{\underline{\hat{\alpha}}}_{\underline{\omega}} \left( \left[ U \right]^{\underline{\hat{\alpha}}}_{\underline{\omega}} \right)^{\dagger} = \left( \left[ U \right]^{\underline{\hat{\alpha}}}_{\underline{\omega}} \right)^{\dagger} \left[ U \right]^{\underline{\hat{\alpha}}}_{\underline{\omega}} = \mathbb{1}_{\mathcal{H}_0} \quad \quad \text{(Schur orthogonality in $\mathcal{H}_0$)}.  
\end{align*}
Unitarity in this context is basically expressing the isometric (and co-isometric) equivalence of the natural bases. The latter result will be useful in our next task.
\end{proof}

\begin{prop}[Path connectedness]\label{pathconnected}
The ground state space $\mathcal{H}_0$ is \textit{path connected}\footnote{Another well known connectedness construction can also be achieved by means of the bases $\left\{ \ket{\underline{\omega}_{\hat{0}}} \right\}_{\underline{\omega} \in H^0}$ (or $\left\{\ket{\underline{\hat{\alpha}}^0}\right\}_{\underline{\hat{\alpha}} \in H_0}$) and the Gram-Schmidt orthogonalization procedure.}.   
\end{prop}
\begin{proof}
We will show that, path connectedness follows Prop. \ref{Uequivalence}. In the following, we will use a standard construction that briefly summarize here. If we write the projection $\Psi^{\underline{\omega}} := \Pi^{\underline{\omega}} \ket{\Psi}$ for $\underline{\omega} \in H^0$ and $\Psi_{\underline{\hat{\nu}}} := \Pi_{\underline{\hat{\nu}}} \ket{\Psi}$ for $\underline{\nu} \in H_0$ , by completeness a general ground state $\ket{\Psi} \in \mathcal{H}_0$ can be decomposed as:
\begin{align*}
\ket{\Psi} = \sum_{\underline{\omega} \in H^0} \Psi^{\underline{\omega}} \ket{\underline{\omega}_{\hat{0}}} = \sum_{\underline{\hat{\alpha}} \in H_0} \Psi_{\underline{\hat{\alpha}}} \ket{\underline{\hat{\alpha}}^0} \quad \text{where the projectors transform as} \quad \Psi_{\underline{\hat{\nu}}} = \left[ U \right]^{\underline{\hat{\alpha}}}_{\underline{\omega}} \Psi^{\underline{\omega}} \quad .
\end{align*}
Thus, the expression: $ \quad \quad \quad \quad  \Psi^{\underline{\omega}} \left( p \right) := \left( 1 - p \right)\Psi^{\underline{\omega}} + p \, I_{\mathcal{H}_0 \times \mathcal{H}_0} \quad \quad$ for  $\quad \quad p \in \left[ 0 , 1 \right] \quad \;\;\; $, \\ is a convex path connecting the identity $I_{\mathcal{H}_0 \times \mathcal{H}_0}$ and the projection $\Psi^{\underline{\omega}}\,$. At the same time, since $\left[ U \right]^{\underline{\hat{\alpha}}}_{\underline{\omega}}$ is a unitary matrix, it can be diagonalized by some matrix $S$ in the form:
\begin{align*}
\left[ U \right]^{\underline{\hat{\alpha}}}_{\underline{\omega}} = S^{-1} D S \quad \text{where} \;\; D \;\; \text{is now a diagonal matrix with eigenvalues in $S^1$} \;\;.
\end{align*}
Now, we can path-connect $D$ to $I_{\mathcal{H}_0 \times \mathcal{H}_0}$ by rotating each element of the diagonal $\lambda_j$, i.e. each eigenvalue, around the unit circle to $1$ in the standard way:
\begin{align*}
S^1 \ni \frac{\lambda_j}{\left( 1 - q \right) + q \left| \lambda_j \right|} \quad \text{for} \quad q \in \left[ 0, 1 \right] \quad , \quad \text{we call this path} \quad D \left( q \right) \quad .  
\end{align*}
Therefore, $\left[ U \right]^{\underline{\hat{\alpha}}}_{\underline{\omega}} \left( q \right) = S^{-1} D \left( q \right) S$ connects the identity $I_{\mathcal{H}_0 \times \mathcal{H}_0}$ to the original matrix $\left[ U \right]^{\underline{\hat{\alpha}}}_{\underline{\omega}}$ without passing through $0$. Finally, the composite:
\begin{align}
\Psi_{\underline{\hat{\alpha}}} \left( t \right)= \left[ U \right]^{\underline{\hat{\alpha}}}_{\underline{\omega}} \left( t \right) \Psi^{\underline{\omega}} \left( t \right) \quad \text{for} \quad t \in \left[ 0, 1 \right] \quad \text{connects any} \quad \Psi_{\underline{\hat{\alpha}}} \quad \text{to} \quad I \quad,   
\end{align} 
which completes the proof.
\end{proof}

Now, we will construct a suitable open cover for $\mathcal{H}_0$ using the operators studied so far.
\begin{prop}[Open Cover] \label{opencover}
The sets $\left\{ U^{\underline{\beta}} \bigcap U_{\underline{\hat{\alpha}}}\right\}_{\underline{\hat{\alpha}} \in H_0 \, , \, \underline{\beta} \in H^0}$ conform an open cover of the ground state subspace $\mathcal{H}_0$. 
\end{prop}
\begin{proof}
Notice that, by completeness, we have that the projectors (\ref{completeGSS}) also satisfy the expressions $\sum_{\underline{\hat{\alpha}} \in H_0} \sum_{\underline{\beta} \in H^0} \Pi_{\underline{\hat{\alpha}}}\Pi^{\underline{\beta}} = \mathbb{1}$ and $\sum_{\underline{\hat{\alpha}} \in H_0} \sum_{\underline{\beta} \in H^0} \Pi^{\underline{\beta}} \Pi_{\underline{\hat{\alpha}}} = \mathbb{1}$. Where, it is straightforward to show that $\Pi^{\underline{\beta}} \Pi_{\underline{\hat{\alpha}}}$ and  $\Pi_{\underline{\hat{\alpha}}}\Pi^{\underline{\beta}}$ satisfy:
\begin{align}
\Pi_{\underline{\hat{\alpha}}}\Pi^{\underline{\beta}} =  \frac{\bar{\chi}_{\underline{\hat{\alpha}}} \left( \underline{\beta} \right) }{\left| H^0 \right|^{\frac{1}{2}}} \ket{\underline{\hat{\alpha}}^0} \otimes \bra{\prescript{}{\hat{0}}{\underline{\beta}}} = \left( \frac{\chi_{\underline{\hat{\alpha}}} \left( \underline{\beta} \right) }{\left| H^0 \right|^{\frac{1}{2}}} \ket{\underline{\beta}_{\hat{0}}} \otimes \bra{\prescript{0}{}{\underline{\hat{\alpha}}}} \right)^{\dagger} = \left( \Pi^{\underline{\beta}} \Pi_{\underline{\hat{\alpha}}} \right)^{\dagger} \quad .
\end{align}
Therefore, it follows that the operators:
\begin{align}\label{partition1}
\Theta_{\underline{\hat{\alpha}}}^{\underline{\beta}} := \frac{1}{2} \left(  \Pi^{\underline{\beta}} \Pi_{\underline{\hat{\alpha}}} + \Pi_{\underline{\hat{\alpha}}}\Pi^{\underline{\beta}} \right) \quad \quad \quad \text{and} \quad \quad \quad \tilde{\Theta}_{\underline{\hat{\alpha}}}^{\underline{\beta}}:= \frac{\Pi_{\underline{\hat{\alpha}}} + \Pi^{\underline{\beta}}}{\left| H^0 \right|+\left| H_0 \right|}
\end{align}
are self adjoint, and satisfy:
\begin{align}\label{partition2}
\sum_{\underline{\hat{\alpha}} \in H_0} \sum_{\underline{\beta} \in H^0} \Theta_{\underline{\hat{\alpha}}}^{\underline{\beta}} = \mathbb{1} \quad \quad \quad , \quad \quad \quad \sum_{\underline{\hat{\alpha}} \in H_0} \sum_{\underline{\beta} \in H^0} \tilde{\Theta}_{\underline{\hat{\alpha}}}^{\underline{\beta}} = \mathbb{1} \quad .
\end{align}
In other words, each set of operators $\left\{ \Theta^{\underline{\beta}}_{\underline{\hat{\alpha}}}\right\}_{\underline{\hat{\alpha}},\underline{\beta}}$ and $\left\{ \tilde{\Theta}^{\underline{\beta}}_{\underline{\hat{\alpha}}}\right\}_{\underline{\hat{\alpha}},\underline{\beta}}$ conform a \textit{positive operator valued measure} (POVM) \cite{Chefles_2000,Peres:2002wx,Ghiloni_2017} for $\underline{\beta}
\in H^0 , \underline{\hat{\alpha}} \in  H_0$ \footnote{This behavior also hints at the presence of a monoidal category. The braiding will  be explored on future works.}. Now, consider a general ground state $\ket{\Psi} \in \mathcal{H}_0$, the application $\Theta_{\underline{\hat{\alpha}}}^{\underline{\beta}} \ket{\Psi} \in U^{\underline{\beta}} \bigcap U_{\underline{\hat{\alpha}}}$, similarly for the $\tilde{\Theta}_{\underline{\hat{\alpha}}}^{\underline{\beta}}$ operator. It follows that the open sets $\left\{ U^{\underline{\beta}} \bigcap U_{\underline{\hat{\alpha}}}\right\}_{\underline{\hat{\alpha}} \in H_0 \, , \, \underline{\beta} \in H^0}$ conform an open cover of $\mathcal{H}_0$.
\end{proof}

Even though we have $\left[ \Theta^{\underline{\beta}}_{\underline{\hat{\alpha}}} , \Pi^0_{\hat{0}} \right] = \left[ \tilde{\Theta}^{\underline{\beta}}_{\underline{\hat{\alpha}}} , \Pi^0_{\hat{0}} \right]=0 \;$ by construction, the operators (\ref{partition1}) do not yet conform a complete set of commuting observables. Basically, we have all the ingredients to build the classifying space. Furthermore, the previous discussion hints to a classification characterized by some structure depending on the jointed elements of $H^0$ and $H_0$. We are now in position to fit the picture we have been constructing in the language of Classifying spaces (See App. \ref{App:classifying}):

\begin{prop}[trivial $H^0 \times H_0$-subspaces]\label{trivialH0H0}
The ground state subspace $\mathcal{H}_0$ admits trivial $H^0 \times H_0$-subspaces.
\end{prop}

\begin{proof}
Consider the following identifications:
\begin{itemize}
\item Take $E = \mathcal{H}_0$ ,
\item Take $B = \bigcup_{\underline{\hat{\alpha}} \in H_0} \bigcup_{\underline{\beta} \in H^0} U^{\underline{\beta}} \bigcap U_{\underline{\hat{\alpha}}}$ with the a newly defined label $b := \left( \underline{\beta} , \underline{\hat{\alpha}} \right) \,$ , 
\item Take $\left. p \right|_{V_b} = \Theta_{\underline{\hat{\alpha}}}^{\underline{\beta}}$ so that $V_b = U^{\underline{\beta}} \bigcap U_{\underline{\hat{\alpha}}}$ , 
\item Take $f_b = f_{\hat{\underline{\alpha}}}^{\underline{\beta}}$ as the local trivializations $f_b \left( \ket{\Psi}g \right) = f_b \left( \ket{\Psi} \right)g $ , $g \in H^0 \times H_0$ ,
\item Finally, take $r_b = Q_{\underline{\hat{\alpha}}} \otimes P^{\underline{\beta}} $ which is a continuous right action by addition on their respective bases (see equation (\ref{def:PQoperators})) .
\end{itemize}
\end{proof}

\begin{prop}[principal $H^0 \times H_0$-bundle]\label{principalH0H0}
The ground state subspace $\mathcal{H}_0$ has a principal $H^0 \times H_0$-bundle structure.
\end{prop}

\begin{proof}
From Prop.\ref{trivialH0H0}, it is immediate to construct the following principal $H^0 \times H_0$-bundle in the following way:
\begin{itemize}
\item Since the set $\left\{ U^{\underline{\beta}} \bigcap U_{\underline{\hat{\alpha}}} \right\}_{\underline{\hat{\alpha}} \in H_0 , \underline{\beta} \in H^0 }$ is an open cover of $\mathcal{H}_0$ (see Prop. \ref{opencover}), $\mathcal{H}_0$ becomes a trivial $H^0 \times H_0$-space ,

\item Given $\mathcal{H}_0$ to be locally trivial, $\mathcal{H}_0 \rightarrow \mathcal{H}_0 / \left( H^0 \times H_0 \right)$ is indeed a principal $H^0 \times H_0$-bundle.
\end{itemize}
\end{proof}

The main result of this paper follows:

\begin{thm}\label{classthm}
The ground state space $\mathcal{H}_0$ is classified by the topological group $H^0 \times H_0$.
\end{thm}

\begin{proof}
We interpret the $\mathcal{H}_0$ structure in the language of classifying spaces:

\begin{itemize}
\item By virtue of the discussion in Prop. \ref{pathconnected} , we have that $\mathcal{H}_0$ is path-connected. Equivalently, $\mathcal{H}_0$ is weakly contractible ,
\end{itemize}

This observation allows us to:

\begin{itemize}
\item  Take $EG = \mathcal{H}_0$ and $BG = \mathcal{H}_0 / \left( H^0 \times H_0 \right)$ ,

\item Since $H^0 \times H_0$ is a discrete group, there exists a $CW$-complex $K\left( H^0 \times H_0 , 1 \right)$ (See App. \ref{app_CW}) such that the first homotopy group $\pi_1 K\left( H^0 \times H_0 , 1 \right) = H^0 \times H_0$ and all other homotopy groups of $K\left( H^0 \times H_0 , 1 \right)$ vanish. 
\end{itemize}
which proves our result.
\end{proof}

Rephrasing, $K\left( H^0 \times H_0 , 1 \right)$ characterizes $H^0 \times H_0$ up to homotopy equivalence. Even though, this last construction is highly non-trivial, it is a standard result in algebraic topology (we refer to a classic book like \cite{Hatcher} for the interested reader).

\subsection{\label{sec:GNSinsights} Some insights from the GNS construction approach.}

In order to motivate the identification of the superselection sectors of these models, we make use of the complementary GNS construction \cite{Balachandran:2013hga,GNS2018} perspective. As to avoid any confusion, we will differentiate between states used in the sense of the previous sections, i.e. merely as elements $\ket{\psi}$ of $\mathcal{H}$, and states in the technical sense of the GNS formalism (see App. \ref{app_GNS} for express definitions, or   \cite{1980applications,bratteli1987operator} for a more in depth treatment of the subject), which we will call these \textit{gns}-states. Naively, for the case of the Hamiltonian operator (\ref{def:Hamiltonian}), the condition (\ref{def:func}) is equivalent to that of the frustration-free condition ($\forall \;\, x \in K$):
\begin{align*}
\varphi \left( A_{\left(x_* \hat{0} \right)x_*} \right) = \varphi \left( B^{\left(x_* 0 \right)x_*} \right) = 1 \;\; , \;\; e = \left\{ \, \underline{\omega} \in H^0 \, , \, \underline{\hat{\nu}} \in H_0 \, \right\} \;\; \text{for} \;\; \mathcal{H}_{\varphi} = \left\{ \mathcal{H} , \hat{\mathcal{H}} \right\} \;\;, \;\; \text{resp.}  \quad .
\end{align*}
Compare the latter with the characterization (\ref{def:GS}). Equivalently, from (\ref{def:GSproj}) and (\ref{GSproj1}), any \textit{gns}-ground state is distinguished by its behavior under the projector $\Pi_{\hat{0}}^0$. Hence, its \textit{gns}-characterization is given by $1 = \varphi \left( \Pi_{\hat{0}}^0 \right)$
for all $\varphi \in K$. Following Doplicher-Haag-Roberts (DHR) \cite{Doplicher:1971wk,Doplicher:1973at} in the context of algebraic quantum field theory, from a vacuum state one could recover all physically relevant properties of the charges, in accordance with the global gauge group by using a specific physically motivated superselection criterion. This has been carried out into the 2D quantum double models succesfully (See, for instance, \cite{NAAIJKENS_2011,Fiedler_2015}. Analogously, the role of the vacuum is played by the translation invariant frustration free ground state. However, the discussion presented in the previous subsection underpins a structure that combines gauge configuration and gauge representation in a joint way. Furthermore, if $e = \left\{ \, \underline{\omega} \in H^0 \, , \, \underline{\hat{\nu}} \in H_0 \, \right\}$ is a generator, so is:
\begin{align}\label{joint}
e := b \in H^0 \times H_0 \quad \text{with} \quad  \ket{b} := \ket{\underline{\omega}_{\hat{0}}} \otimes \ket{\underline{\hat{\nu}}^0} \quad \text{and} \quad \pi_{\varphi} \left( \Pi^{0}_{\hat{0}} \right) =\Pi^{0}_{\hat{0}} \otimes \Pi^{0}_{\hat{0}} \quad , 
\end{align} 
although this choice is not unique. This new basis is also orthonormal and, even more important, maximal since the following injections are natural:
\begin{align}
\ket{\underline{\omega}_{\hat{0}}} \hookrightarrow \ket{\underline{\omega}_{\hat{0}}} \otimes \ket{\underline{\hat{0}}^0} \quad , \quad \ket{\underline{\hat{\nu}}^0} \hookrightarrow \ket{\underline{0}_{\hat{0}}} \otimes \ket{\underline{\hat{\nu}}^0} \quad .     
\end{align}
By virtue of (\ref{joint}), the functional (\ref{def:func}) admits the \textit{gns}-states maximal superposition:
\begin{align}\label{convex:sup}
\varphi \left( \mathcal{O} \right) = \sum_{b \in H^0 \times H_0} \lambda_{b} \, \mathcal{O}_{b} \;\; , \;\; \text{where} \;\; \mathcal{O}_{b} := \text{Tr}_{H^0 \times H_0} \pi_{\varphi} \left( \mathcal{O} \right) \hat{\Pi}_b \;\; \text{with} \;\; \hat{\Pi}_b := \ket{b} \otimes \bra{b} \quad ,
\end{align} 
which are pure \textit{gns}-states. Notice that the ground state space projector $\Pi_{\hat{0}}^0$ acting on pure \textit{gns}-states yields $\left( \Pi_{\hat{0}}^0 \right)_b = \text{Tr}_{H^0 \times H_0} \left( \pi_{\varphi} \left( \Pi_{\hat{0}}^0 \right) \circ \hat{\Pi}_b \right) = \text{Tr}_{H^0 \times H_0} \left( \hat{\Pi}_b \right)  = 1$. Thus, the decomposition (\ref{convex:sup}) satisfies the condition:
\begin{align}\label{convex}
\varphi \left( \Pi_{\hat{0}}^0 \right) = \sum_{b \in H^0 \times H_0} \lambda_{b} \left( \Pi_{\hat{0}}^0 \right)_{b} = \sum_{b \in H^0 \times H_0} \lambda_b = 1 \quad ,
\end{align}
i.e. the \textit{gns}-ground state space $K$ defined in (\ref{gnsGSS}) is \textit{closed} and \textit{convex} and consists of $\left| H_0 \right|\left| H^0 \right| $ pure \textit{gns}-states \cite{zalinescu2002convex,lucchetti2005convexity}. We can consider this result to be the $n$-dimensional \textit{higher abelian gauge} version of that appearing in Theorem 2.2 of \cite{Cha2017}. Moreover, since the pure gns-states are in one-to-one correspondence with the elements of the classifying space of $\mathcal{H}_0$, the pure gns-states constructed from the generators (\ref{joint}) are a suitable representation of the superselection basis. Furthermore, the convergence of the sum (\ref{convex}) allows the following string of equalities to be satisfied:
\begin{align*}
1 = \sum_{b \in H^0 \times H_0} \lambda_{b} \text{Tr}_{H^0 \times H_0} \left( \hat{\Pi}_b \right) = \text{Tr}_{H^0 \times H_0} \left( \sum_{b \in H^0 \times H_0} \lambda_{b} \hat{\Pi}_b \right) \quad \quad .
\end{align*}
This is, the set of numbers $\left\{ \lambda_b \right\}_{b \in H^0 \times H_0}$ can be interpreted as probabilities defining the mixed density operator of the ground state space $\mathcal{H}_0$ :
\begin{align}\label{density}
\hat{\rho}_{H^0 \times H_0} := \sum_{b \in H^0 \times H_0} \lambda_b \hat{\Pi}_b  \quad \quad ,
\end{align}
associated to a \textit{gns}-ground state $\varphi \in K$. Notice that these do not evolve with time, as it is attested by the Von-Neumann equation:
\begin{align*}
\frac{\partial \hat{\rho}_{H^0 \times H_0}}{\partial t} = \frac{1}{ih} \left[ \pi_{\varphi} \left( H \right) , \hat{\rho}_{H^0 \times H_0} \right] = - \frac{1}{ih} \left[ \ln \left( \pi_{\varphi} \left( \Pi_{\hat{0}}^0 \right) \right) , \hat{\rho}_{H^0 \times H_0} \right] \quad ,
\end{align*}
where we have used  (\ref{def:Hamiltonian}) and the fact that $\left[ \pi_{\varphi} \left( \Pi_{\hat{0}}^0 \right) , \hat{\Pi}_b \right]=0$. Moreover, the projectors $\hat{\Pi}_b$ are recognized as the measuring operators for the different sectors of the ground state space and become the prototype for studying the \textit{superselection sectors} in future works.

To complete the picture, writing explicitly $\lambda_b = \lambda_{\underline{\mu}}^{\underline{\hat{\nu}}}$ for $b = \left( \underline{\mu} , \underline{\hat{\nu}} \right) \in H^0 \times H_0$, we show that the description of the ground state given in Sec. \ref{sec:topdesc} is recovered when performing the following partial traces:
\begin{align}
\nonumber \hat{\rho}_{H_0} & := \text{Tr}_{H^0} \left( \hat{\rho}_{H^0 \times H_0} \right) = \sum_{\underline{\hat{\nu}} \in H_0} \Lambda^{\hat{\nu}} \Pi_{\underline{\hat{\nu}}} \quad \;\; \text{with} \quad \, \Lambda^{\underline{\hat{\nu}}} := \sum_{\underline{\mu} \in H^0} \lambda_{\underline{\mu}}^{\underline{\hat{\nu}}} \quad \text{and} \quad \sum_{\underline{\hat{\nu}} \in H_0} \Lambda^{\underline{\hat{\nu}}} = 1 \quad , \\
\label{ancilla} \hat{\rho}_{H^0} &:= \text{Tr}_{H_0} \left( \hat{\rho}_{H^0 \times H_0} \right) = \sum_{\underline{\omega} \in H^0} \Lambda_{\underline{\omega}} \Pi^{\underline{\omega}} \quad \text{where} \quad \Lambda_{\underline{\omega}} := \sum_{\underline{\hat{\nu}} \in H_0} \lambda_{\underline{\omega}}^{ \underline{\hat{\nu}}} \quad \text{and} \quad \sum_{\underline{\omega} \in H^0} \Lambda_{\underline{\omega}} = 1 \quad .
\end{align}

In terms of quantum computing, considering the unitary equivalence of the $\left\{ \ket{\underline{\omega}_{\hat{0}}} \right\}_{\underline{\omega} \in H^0}$ and $\left\{ \ket{\underline{\hat{\nu}}^0} \right\}_{\underline{\hat{\nu}} \in H_0}$ bases shown in subsection \ref{subsec:pathconnec}, plus the set of relations (\ref{ancilla}), it follows that become \textit{ancilla} to one another \cite{https://doi.org/10.48550/arxiv.quant-ph/0305068,watrous_2018}. This last observation explains why only $\left| H^0 \right|$, or equivalently $\left| H_0 \right|$, appear in the calculation of the GSD.

\section{\label{sec:FinalRemarks}Final Remarks}

The main result of this paper is the recognition of the topological group $H^0 \times H_0$ as the classifying group of the ground state subspace $\mathcal{H}_0$ for the abelian higher gauge symmetry models defined in \cite{dealmeida2017topological,Ibieta_Jimenez_2020}. A proof based on the GNS formalism has also been developed and further verifies the results shown here. The latter will be submitted to peer review shortly. 

Regarding the topological contents of these models, if we allow ourselves to hypothesize a bit, analogous arguments of those leading to Prop. \ref{restrictionGS}, shows that we can equivalently consider quotient spaces:
\begin{align}\label{totalclassispace}
\underline{\Omega} := \text{hom}^0 / \sim^0 \; \cong \; H^0 \oplus \text{im} \left( d^0  \right) / \sim^0 \;\;\text{and}  \;\; \underline{\hat{P}}:= \text{hom}_0 / \sim_0 \; \cong \; H_0 \oplus \text{im} \left( d_1  \right) / \sim_0 \quad ,  
\end{align}
as the space of \textit{gauge equivalent configuration} and \textit{co-gauge equivalent representations}, respectively. 
Similar arguments to those leading to the classifying space of the ground state subspace $\mathcal{H}_0$ can be generalized to include the entire Hilbert space $\mathcal{H}$. If this is the case (modulo the conditions of subsection \ref{subsec:pathconnec} are satisfied), the classifying set $\mathcal{H}$ should be given by $\underline{\Omega} \otimes \underline{\hat{P}}$ up to isomorphisms. If we consider the decompositions (\ref{totalclassispace}), the following sectors are obtained:  
\begin{enumerate}
    \item $ H^0 \otimes H_0 $: is recognized as the classifying group for the ground state subspace $H_0$. This group is topological and has been properly identified in this paper;
    
    \item $H^0 \otimes \left( \text{im} d_1 / \sim_0 \right) $ and $ \left( \text{im} d^0 / \sim^0  \right) \otimes H_0$: can be recognized as characterizing \textit{semi-topological excited states}. We refer to these sectors as \textit{semi-topological} because of the presence of $H^0$ and $H_0$, while the sets $\text{im} d_1 / \sim_0$ and $\text{im} d^0 / \sim^0$, being in principle isomorphic to one another, need further topological structure analysis. Notice also that these spaces should be related by a braided monoidal category;
    
    \item $\left( \text{im} d^0 / \sim^0  \right) \otimes \left( \text{im} d_1 / \sim_0 \right) $: is recognized as characterizing \textit{pure excited states}.
\end{enumerate}

It remains to study the \textit{semi-topological} excited states (the ones we expect to be classified by $2$) which we believe are the subspaces in which anyons can appear. Some preliminary results have also been obtained by means of a GNS representation formalism and we expect to present these shortly.

\begin{acknowledgments}
The author wishes to thank ANID through its support through the grant FONDECYT No. 11241170 along the complation of this paper. The authors have no conflicts to disclose.
\end{acknowledgments}

\appendix


\section{\label{App:classifying} Review on Classifying spaces}

We briefly review the ingredients of a classifying space so we can immediately identify that the ground state subspace $\mathcal{H}_0$ is indeed characterized by one. A principal $G$-bundle is a local trivial fibration where the charts are compatible with a group action. More concretely, for $G$ a topological group, a principal (right) $G$-bundle is the triple $\left( E , B , p \right)$, where $p: E \rightarrow B$ is a map, together with a continuous right action $r : E \times G \rightarrow E$ satisfying:
\begin{enumerate}[i)]
\item For all $x \in E$ and $g \in G$ we have $p(xg) = p(x)$.
\item For every $b \in B$ there is an open neighbourhood $V_b$ and a $G$-homeomorphism $f_b: p^{-1} \left( V_b \right) \rightarrow V_b \times G$ (where $G$ acts on $p^{-1}\left( V_b \right)$ by restriction of $r$ and acts on $V_b \times G$ by $\left( \left( v , x \right) , g \right) \rightarrow \left( v ,  xg \right)$) such that the following diagram commutes:
\begin{figure*}[!ht]
\begin{tikzcd}
p^{-1}\left( V_b \right) \arrow{r}{f_b} \arrow[swap]{d}{\left.p\right|_{V_b}} & V_b \times G \arrow{dl}{\text{pr}_1} \\%
V_b & 
\end{tikzcd}
\end{figure*}
\end{enumerate}
The action of $G$ on $p^{-1}\left( V_b \right)$ is well defined by the first condition. By the same reason, the map $p$ induces a map $p: E/G \rightarrow B$. The second condition implies that $G$ acts freely on $E$ and that the map $p: E/G \rightarrow B$ is a homeomorphism. A right $G$-space is said to be locally trivial if it has an open covering by trivial $G$-subspaces. If $E$ is locally trivial, then $E \rightarrow E/G$ is a principal $G$-bundle.

\begin{figure}[h!]
\centering
\tikzset{every picture/.style={line width=0.75pt}} 
\begin{tikzpicture}[x=0.75pt,y=0.75pt,yscale=-1,xscale=1]

\draw  [dash pattern={on 1.69pt off 2.76pt}][line width=1.5]  (34,210) .. controls (41,193.5) and (117,183.5) .. (124,210) .. controls (131,236.5) and (150,253) .. (131,279.5) .. controls (112,306) and (54,300) .. (34,270) .. controls (14,240) and (27,226.5) .. (34,210) -- cycle ;
\draw  [dash pattern={on 1.69pt off 2.76pt}][line width=1.5]  (94,221) .. controls (102,214.5) and (116,220.5) .. (143,225.5) .. controls (170,230.5) and (196,243.5) .. (184,281) .. controls (172,318.5) and (103,314.5) .. (83,284.5) .. controls (63,254.5) and (86,227.5) .. (94,221) -- cycle ;
\draw [line width=1.5]    (85,261.5) .. controls (103.84,247.37) and (92.61,265.3) .. (107,255.5) .. controls (121.39,245.7) and (105.84,265.37) .. (127,249.5) ;
\draw [line width=1.5]    (53,78) .. controls (73.09,62.93) and (114.24,120.23) .. (162.12,79.87) .. controls (210,39.5) and (250.09,96.43) .. (270,81.5) ;
\draw  [dash pattern={on 4.5pt off 4.5pt}]  (64,7.5) -- (81,152.5) ;
\draw  [dash pattern={on 4.5pt off 4.5pt}]  (99,16.5) -- (100,150.5) ;
\draw  [dash pattern={on 4.5pt off 4.5pt}]  (135,22.5) -- (125,157.5) ;
\draw  [dash pattern={on 4.5pt off 4.5pt}]  (178,35.5) -- (164,160.5) ;
\draw  [dash pattern={on 4.5pt off 4.5pt}]  (158,34.5) -- (143,162.5) ;
\draw  [dash pattern={on 4.5pt off 4.5pt}]  (213,34.5) -- (188,160.5) ;
\draw  [dash pattern={on 4.5pt off 4.5pt}]  (246,29.5) -- (213,163.5) ;
\draw [line width=1.5]    (131,95.5) .. controls (61.42,157.24) and (88.84,218.02) .. (101.26,256.19) ;
\draw [shift={(102,258.5)}, rotate = 252.47] [color={rgb, 255:red, 0; green, 0; blue, 0 }  ][line width=1.5]    (14.21,-4.28) .. controls (9.04,-1.82) and (4.3,-0.39) .. (0,0) .. controls (4.3,0.39) and (9.04,1.82) .. (14.21,4.28)   ;
\draw    (10,207.5) .. controls (18.41,199.09) and (21.91,186.96) .. (54,187.5) .. controls (86.09,188.04) and (137.95,185.92) .. (173,199.5) .. controls (208.05,213.08) and (259,234) .. (273,254.5) ;
\draw  [fill={rgb, 255:red, 0; green, 0; blue, 0 }  ,fill opacity=1 ] (126,95.5) .. controls (126,93.15) and (128.01,91.25) .. (130.5,91.25) .. controls (132.99,91.25) and (135,93.15) .. (135,95.5) .. controls (135,97.85) and (132.99,99.75) .. (130.5,99.75) .. controls (128.01,99.75) and (126,97.85) .. (126,95.5) -- cycle ;
\draw  [fill={rgb, 255:red, 0; green, 0; blue, 0 }  ,fill opacity=1 ] (97.5,258.5) .. controls (97.5,256.15) and (99.51,254.25) .. (102,254.25) .. controls (104.49,254.25) and (106.5,256.15) .. (106.5,258.5) .. controls (106.5,260.85) and (104.49,262.75) .. (102,262.75) .. controls (99.51,262.75) and (97.5,260.85) .. (97.5,258.5) -- cycle ;

\draw (33,222.4) node [anchor=north west][inner sep=0.75pt]    {$\mathnormal{U}_i$};
\draw (150,255.4) node [anchor=north west][inner sep=0.75pt]    {$U_j$};
\draw (96,219.4) node [anchor=north west][inner sep=0.75pt]    {$V_{b}$};
\draw (80,262) node [anchor=north west][inner sep=0.75pt]    {$\gamma $};
\draw (30,81.4) node [anchor=north west][inner sep=0.75pt]    {$\tilde{\gamma }$};
\draw (247,97.4) node [anchor=north west][inner sep=0.75pt]    {$p^{-1}( V_{b})$};
\draw (68,163.4) node [anchor=north west][inner sep=0.75pt]    {$p$};
\draw (116,74.4) node [anchor=north west][inner sep=0.75pt]    {$x$};
\draw (238,252.4) node [anchor=north west][inner sep=0.75pt]    {$B$};
\draw (99.5,261.9) node [anchor=north west][inner sep=0.75pt]    {$p( x)$};
\end{tikzpicture}
\caption{\label{gbundle} The picture shows a geometric representation of a principal $G$-bundle. $\gamma$ is a path in $V_b = U_i \cap U_j$ with $U_i$, $U_j$ elements of an open cover, while $\tilde{\gamma}$ is a horizontal lift on $p^{-1} \left( V_b \right)$. The straight dashed lines crossing $\tilde{\gamma}$ represent the $G$ structure of the section}
\end{figure}

A classifying space $BG$ of a topological group $G$ is then constructed from the quotient of a weakly contractible space $EG$\footnote{Weakly contractible means a topological space with all of its homotopy groups being trivial.} by a proper free action of $G$. Finally, a principal $G$-bundle $p: EG \rightarrow BG$ is called universal if it is numerable trivial and if for every numerable trivial principal $G$-bundle $q : E \rightarrow B$ there exist up to homotopy a unique bundle map from $q$ to $p$ \cite{may1999concise}.


\section{\label{app_CW}A brief on \textit{CW}-complex of the $K\left( G , n \right)$ type.}

Let write $D^n$ to be the $n$-dimensional closed unit ball and $S^{n-1}$ the $(n -1)$-dimensional unit sphere (or the boundary of $D^n$). A CW-complex is a topological space $X$ and a collection of continuous maps $\phi^n_{\alpha}: D^n \rightarrow  X$, called
characteristic maps, and $e^n_{\alpha} = \phi^n_{\alpha}\left( \text{Int}\left(D^n\right)\right)$ called $n$-cells of $X$ obeying the following  properties:
\begin{enumerate}[i)]
    \item $X= \bigcup_{n \geq 0 \, , \, \forall \alpha} e^n_{\alpha} \quad $ ,
    \item $e^n_{\alpha} \bigcap e^m_{\beta} = \empty$ unless $\left.\phi^n_{\alpha}\right|_{\text{Int} \left( D^n \right)}$ is a homeomorphism and $\alpha = \beta$ and $n=m \quad$,
    \item An $n$-skeleton $X^n = \bigsqcup_{0 \geq i \geq n \, , \, \forall \alpha} e^i_{\alpha} \quad $.
\end{enumerate}

This construction exhibits:
\begin{itemize}
    \item \textit{Closure finiteness (C)}: The closure of each cell of $X$ is contained in finitely many other cells; and
    \item \textit{Weak topology (W)}: $A \subset X$ is open or closed if and only if $A \cap X^n$ is open or closed for all $n$.
\end{itemize} 

Hence, the name \textit{CW}-complex. We refer to a subcomplex of a \textit{CW}-complex $X$ as a union $A$ of cells in $X$ such that the closure of each cell is also contained in $A$, in other words, $A$ is also a \textit{CW}-complex. A $CW$-pair $(X, A)$ is \textit{CW}-complex $X$ and subcomplex $A$.

\begin{defn}[$K\left(G ,n \right)$ \textit{CW}-complex] For a group $G$, we say a topological space $X$ is a $K \left( G , n \right)$, or \textit{Eilenberg-MacLane} space, if $\pi_i\left( X \right)$ is isomorphic to $G$ for $i = n$ and trivial otherwise.
\end{defn}

\begin{defn}[$n$-connectedness]
A space $X$ is said to be $n$-connected if $\pi_i \left( X \right) = 0$ for all $i \leq n$. Likewise, a pair $\left( X, A \right)$ is $n$-connected if $\pi_i \left( X, A \right) = 0$ for all $i \leq n$.
\end{defn}

The following are important theorems that will be stated without a proof for the sake of brevity. For a general reference we refer the reader to \cite{Hatcher}.

\begin{thm}[van Kampen]
Let $X$ be a union of path-connected open sets $A_{\alpha}$, each containing the basepoint $x_0$ with the intersections $A_{\alpha} \cap A_{\beta}$ and $A_{\alpha} \cap A_{\beta} \cap A_{\gamma}$ path connected. Let $i_{\alpha,\beta} : \pi_1\left( A_{\alpha} \cap A_{\beta} \right) \rightarrow \pi_1 \left( A_{\alpha} \right)$ be the homomorphism induced by the inclusion map of $A_{\alpha} \cap A_{\beta}$ in $A_{\alpha}$ and let $N$ be the normal subgroup generated by elements of the form $i_{\alpha,\beta} \left( \omega \right) i^{-1}_{\beta,\alpha} \left( \omega \right) \, .$ Then, $\pi_1 \left( X \right)$ is isomorphic to $\star_{\alpha} \pi_1 \left( A_{\alpha} \right) / N \,$, where $\star$ is the free product.
\end{thm}

\begin{thm}[Cellular approximation]
Given two \textit{CW}-complexes $X$ and $Y\,$,
every map $f : X \rightarrow Y$ is homotopic to a map $g : X \rightarrow Y$ with the property that $g \left( X^n \right) \subset Y^n$ for all $n$. A map $g$ with this property is called a cellular map. Moreover, $g$ may be taken to equal $f$ on any subcomplex for which $f$ is already cellular.
\end{thm}

\begin{thm}[Excision] 
Let $X$ be a \textit{CW}-complex which can be decomposed as a union of subcomplexes $A$ and $B$ with their intersection $C$ nonempty and connected. Then, if $\left( A, C \right)$ is $m$-connected and $\left( B, C \right)$ is $n$-connected, the map $\pi_i \left( A, C \right) \rightarrow \pi_i \left( X, B \right)$ induced by inclusion is an isomorphism for $i < n + m$.
\end{thm}

\begin{thm}[Whitehead] 
Let $X$ and $Y$ be \textit{CW}-complexes. If a map $f : X \rightarrow Y$ induces isomorphisms in each homotopy group, then it is a homotopy equivalence.
\end{thm}

The previous theorems ensure the construction of a $K \left( G , n \right)$ \textit{CW}-complex when the base space is path-connected, which is the case for $\mathcal{H}_0$.

\section{\label{app_GNS} Some basic elements of the GNS-construction}

The Gelfand-Naimark-Segal (GNS) construction states that, every $C^*$-algebra is isometrically $*$-isomorphic to a $*$-subalgebra of bounded linear operators on some Hilbert space $\mathcal{H}$. The latter uses a positive linear functional $\varphi$ and the left regular representation $A \rightarrow \text{End}_{\mathbb{C}} \left( A \right)$ to produce a cyclic (irreducible if $\varphi$ is a \textit{pure state}) representation $\varphi : A \rightarrow \mathcal{B} \left( \mathcal{H} \right)$. 

\begin{defn}[$C^*$-algebra]
A $C^*$-algebra $A$ is a unital Banach algebra. This is,  a complete normed linear space over $\mathbb{C}$ with continuous associative multiplication $A \times A \rightarrow A$ and $\left\| 1 \right\| = 1$, along with a map $* : A \rightarrow A$ satisfying (for all $a,b \in A$):
\begin{align*}
& \left( a^* \right)^* = a \quad \quad, & \left( \alpha a + b \right)^* = \bar{\alpha} a^* + b^* \quad \alpha \in \mathbb{C} \quad, \\
& \left( a b \right)^* = b^* a^* \quad \quad , & \left\| a^* a \right\| = \left\| a \right\|^2 \quad \quad .
\end{align*}
\end{defn}

If $\mathcal{H}$ is a Hilbert space, which is our case, then $\mathcal{B} \left( \mathcal{H} \right)$ is known to be a $C^*$-algebra with $*$ beig the usual adjoint defined by $\left\langle x | T y \right\rangle = \left\langle T^* x | y \right\rangle $.

\begin{defn}[$C^*$-algebra representation]
A representation of a $C^*$-algebra is a $*$-homomorphism $\pi : A \rightarrow \mathcal{B} \left( \mathcal{H} \right)$. 
\end{defn}

Two representations $\pi_1 : A \rightarrow \mathcal{B}\left( \mathcal{H}_1 \right)$ , $\pi_2 : A \rightarrow \mathcal{B} \left( \mathcal{H}_2 \right)$, are said to be unitarily equivalent if there is a unitary operator $U : \mathcal{H}_1 \rightarrow \mathcal{H}_2$ such that $ U \pi_1 \left( a \right) = \pi_2 \left( a \right) U$ for all $a \in A$. A representation $\pi$ is said to be cyclic if there exists $e \in \mathcal{H}$ such that $\left\{ \pi \left( a \right) e : a \in A \right\}$ is dense in $\mathcal{H}$. Finally, 
a representation is said to be topologically irreducible if it has no proper, nontrivial closed invariant subspaces.

Quantum mechanically, given a bounded linear functional $\varphi : \mathcal{O} \rightarrow \mathbb{C}$ such that, for $O \in \mathcal{O}$ the set observables, we have $\varphi \left( O \right) \geq 0$ if $O \geq 0$ and $\varphi \left( \mathbb{1} \right) = 1$. The set of all \textit{gns}-states is denoted $\mathcal{O}^*_{+,1}$. Extremal points of $\mathcal{O}^*_{+,1}$ are called pure \textit{gns}-states in this paper. A \textit{gns}-ground state set space is then given by:
\begin{align}\label{gnsGSS}
K = \left\{ \varphi \in \mathcal{O}^*_{+,1} \;\; | \;\; \forall \; O \in \mathcal{O}: \varphi \left( O^{*} \delta \left( O \right)  \right) \geq 0 \right\} \quad ,
\end{align}
where we have used the time evolution (\ref{timeevolution}). The above set is known to be compact and closed in the $\text{weak}^*$ topology. The \textit{gns}-states and the states in the rest of the paper are related in the ground state space by the linear functional:
\begin{align}\label{def:func}
\varphi \left( O \right) := \text{Tr} \bra{e} \pi_{\varphi} \left( O \right) \ket{e} \quad \text{for} \quad O \in \mathcal{O} \quad \text{and} \quad e \;\; \text{a generator of} \;\; \mathcal{H}_{\varphi} \quad ,
\end{align} 
where $\pi_{\varphi}$ is an irreducible representation of the algebra $\mathcal{O}$ (trivial in the configuration basis or its representation analog in our case)\footnote{At finite temperature T, equilibrium states are defined by the KMS-condition \cite{2021arXiv210300901B}} and $\text{Tr}$ is the total trace operator. 

\begin{defn}[Unitarily Equivalent]\label{unitaryGNS}
If $\pi : A \rightarrow \mathcal{H}$ is another cyclic representation with generator $e'$
such that $\varphi \left( a \right) = \text{Tr}_e \bra{e'} \pi \left( a \right) \ket{e'}$, then $\pi$ and $\pi_\varphi$ are unitarily equivalent.
\end{defn}

\begin{lem}
It follows immediately that the bases $\left\{ \ket{\underline{\omega}_{\hat{0}}}\right\}_{\underline{\omega} \in H^0}$ and $\left\{ \ket{\underline{\hat{\alpha}}} \right\}_{\underline{\hat{\alpha}} \in H_0}$ of section \ref{subsec:gaugeeq} are unitarily equivalent in the sense of Def. \ref{unitaryGNS}.
\end{lem}


\bibliographystyle{apsrev4-1}
\bibliography{sample.bib}
\end{document}